\documentclass[12pt]{article}
\usepackage{times}
\usepackage{natbib}
 \bibpunct{(}{)}{;}{a}{,}{,}
\usepackage{graphicx}
\usepackage{amsmath,amssymb,amsthm}
\usepackage{bm}
\usepackage{rotating}
\usepackage[vlined]{algorithm2e}
\usepackage{rotating}
\usepackage{multicol}
\usepackage{titlesec}
\usepackage{latexsym}
\usepackage[pdftex,colorlinks=true,linkcolor=blue,citecolor=blue,urlcolor=blue,bookmarks=false,pdfpagemode=None]{hyperref}
\usepackage{url}
\usepackage{verbatim}
\usepackage{fancyhdr}
\usepackage{setspace}
\usepackage{paralist}
\usepackage{boxedminipage}
\usepackage{color}
\definecolor{lightgrey}{rgb}{0.85,0.85,0.85}
\usepackage{lineno}
\usepackage[top=1in, bottom=1in, left=1in, right=1in]{geometry}

\theoremstyle{plain}

\titlespacing{\section}{0pt}{*0.9}{*0.9}
\titlespacing{\subsection}{0pt}{*0.8}{*0.8}
\titlespacing{\subsubsection}{0pt}{*0.8}{*0.8}

\newcommand{\bv}{\begin{array}}
\newcommand{\ev}{\end{array}}
\newcommand{\bit}{\begin{itemize}}
\newcommand{\eit}{\end{itemize}}
\newcommand{\ben}{\begin{enumerate}}
\newcommand{\een}{\end{enumerate}}
\newcommand{\beq}{\begin{equation}}
\newcommand{\eeq}{\end{equation}}
\newcommand{\bvq}{\begin{eqnarray}}
\newcommand{\evq}{\end{eqnarray}}

\newcommand{\diag}[1]{\mathrm{diag}\left(#1\right)}

\DeclareMathOperator{\Unif}{Unif}

\DeclareMathOperator{\Var}{\mathbb{V}ar}
\DeclareMathOperator{\Cov}{\mathbb{C}ov}

\newtheorem{theorem}{Theorem}

\begin{document}
\pagestyle{empty}

\title{Estimating latent processes on a network\\ from indirect measurements}
\author{ Edoardo M. Airoldi\thanks{Edoardo M.~Airoldi is an Assistant Professor of Statistics at Harvard University, a Principal Investigator at the Broad Institute of MIT \& Harvard, and an Alfred P. Sloan Research Fellow (airoldi@fas.harvard.edu). Alexander W.~Blocker is a PhD candidate in the Department of Statistics at Harvard University (ablocker@fas.harvard.edu).} \and Alexander W. Blocker}
\date{}

\maketitle
\thispagestyle{empty}

\newpage

\begin{abstract}
In a communication network, point-to-point traffic volumes over time are critical for designing protocols that route information efficiently and for maintaining security, whether at the scale of an internet service provider or within a corporation. 
While technically feasible, the direct measurement of point-to-point traffic imposes a heavy burden on network performance and is typically not implemented. 
Instead, indirect aggregate traffic volumes are routinely collected.
We consider the problem of estimating point-to-point traffic volumes, $\bm x_t$, from aggregate traffic volumes, $\bm y_t$, given information about the network routing protocol encoded in a matrix $A$.
This estimation task can be reformulated as finding the solutions to a sequence of {\em ill-posed} linear inverse problems, $\bm y_t= A\, \bm x_t$, since the number of origin-destination routes of interest is higher than the number of aggregate measurements available.

Here, we introduce a novel multilevel state-space model of aggregate traffic volumes with realistic features.
We implement a na\"ive strategy for estimating unobserved point-to-point traffic volumes from indirect measurements of aggregate traffic, based on particle filtering.
We then develop a more efficient two-stage inference strategy that relies on model-based regularization: a simple model is used to calibrate regularization parameters that lead to efficient/scalable inference in the multilevel state-space model.
We apply our methods to corporate and academic networks, where we show that the proposed inference strategy outperforms existing approaches and scales to larger networks.
We also design a simulation study to explore the factors that influence the performance.
Our results suggest that model-based regularization may be an efficient  strategy for inference in other complex multilevel models.\newline

\vfill

\end{abstract}

\singlespacing 

\newpage
\tableofcontents

\onehalfspacing

\newpage

\pagestyle{fancy}
\setcounter{page}{1}

\section{Introduction}
\label{sec:intro}

A pervasive challenge in multivariate time series analysis is the estimation of non-observable time series of interest $\{\bm x_{t} : t = 1 \ldots T \}$ from indirect noisy measurements $\{\bm y_{t} : t = 1 \ldots T \}$, typically obtained through an aggregation or mixing process, $\bm y_t = \bm a (\bm x_t) ~ \forall t$.
The inference problem that arises in this setting is often referred to as an {\em inverse}, or {\em deconvolution}, problem \citep[e.g.,][]{Hans:1998,case:berg:2001,meis:2009} in the statistics and computer science literatures, and qualified as {\em ill-posed} because of the lower dimensionality of the measurement vectors with respect to the non-observable estimands of interest.
Ill-posed inverse problems lie at the heart of a number of modern applications, including image super-resolution and positron emission tomography where we want to combine many 2D images in a 3D image consistent with 2D constraints \citep{shep:krus:1978,vard:shep:kauf:1985}; blind source separation where there are more sound sources than sound tracks (i.e., the measurements) available \citep{Liu:Chen:1995,lee:lewi:giro:sejn:1999,parr:sajd:2003}; and inference on cell values in contingency tables where two-way or multi-way margins are pre-specified \citep{Bish:Fien:Holl:1975,dobr:teba:west:2006}.

We consider a setting in which high-dimensional multivariate time series $\bm x_{1:T}$ mix on a network. 
 Individual time series correspond to traffic directed from a node to another. 
 The aggregation process encodes the routing protocol---whether deterministic of probabilistic---that determines the path traffic from any given source follows to reach its destination.
 This type of mixing can be specified as a linear aggregation process $A$. 
This problem setting leads to the following sequence of ill-posed linear inverse (or deconvolution) problems, 
\begin{equation}
 \label{eq:linearproblem}
 \bm y_t = A \, \bm x_t, \quad \hbox{ s.t. } \bm y_t, \bm x_t \geq0 ~ \hbox{ for } t=1\dots T,
\end{equation}
since the observed aggregate traffic time series are low dimensional, $\bm y_t \in \mathbb{R}^m$, while the latent point-to-point traffic time series of interest are high-dimensional, $\bm x_t \in \mathbb{R}^n$. Thus the matrix $A_{m\times n}$ is rank deficient, $r(A)=m<n$, in this general problem setting.

The application to communication networks that motivates our research is {\em (volume) network tomography}; an application originally introduced by \citet{vardi:1996}, which has quickly become a classic since \citep[e.g., see][]{vand:iann:1994, tebaldiwest1998, Cao:Dav:Van:Yu:2000, Coates:Hero:Nowak:Yu:2002, Medina:2002, zhang:roug:lund:dono:2003, lian:yu:2003, Airo:Falo:2004, castro2004, LakhinaEtAl2004, lawr:2006a, fang:vard:zhan:2007, bloc:airo:2010}.
An established engineering practice is at the root of the inference problem in network tomography.
Briefly, the availability of point-to-point (or origin-destination) traffic volumes over time is critical for 
 reliability analysis (e.g., predicting flows and failures), 
 traffic engineering (e.g., minimizing congestion),
 capacity planning  (e.g., forecasting requirements),
 and security management (e.g., detecting anomalous traffic patterns). 
While technically possible, however, the direct measurement of point-to-point traffic imposes a heavy burden on network performance and is never implemented, except for special purposes over short time periods. 
Instead, indirect aggregate traffic volumes are routinely collected.
As a consequence, network engineers must solve the ill-posed linear inverse problems in Equation \ref{eq:linearproblem} to recover point-to-point traffic.
We give pointers to a vast literature that spans statistics, computer science, and operations research in Section \ref{sec:litreview}.

Figure \ref{fig:nt} provides an illustration of a communication network and the key mathematical quantities in the application to network tomography.
Dashed circles, {\em rs a--e}, represent routers and switches. 
Solid circles, {\em sn} 1--11, represent sub-networks.
 Intuitively, messages are sent from a subnetwork (origin) to another (destination) over the network.
 Routers and switches are special-purpose computers that quickly scan the messages and route them according to a pre-specified routing protocol.
\begin{figure}[t!]
  \centering
   \includegraphics[width=0.75\textwidth]{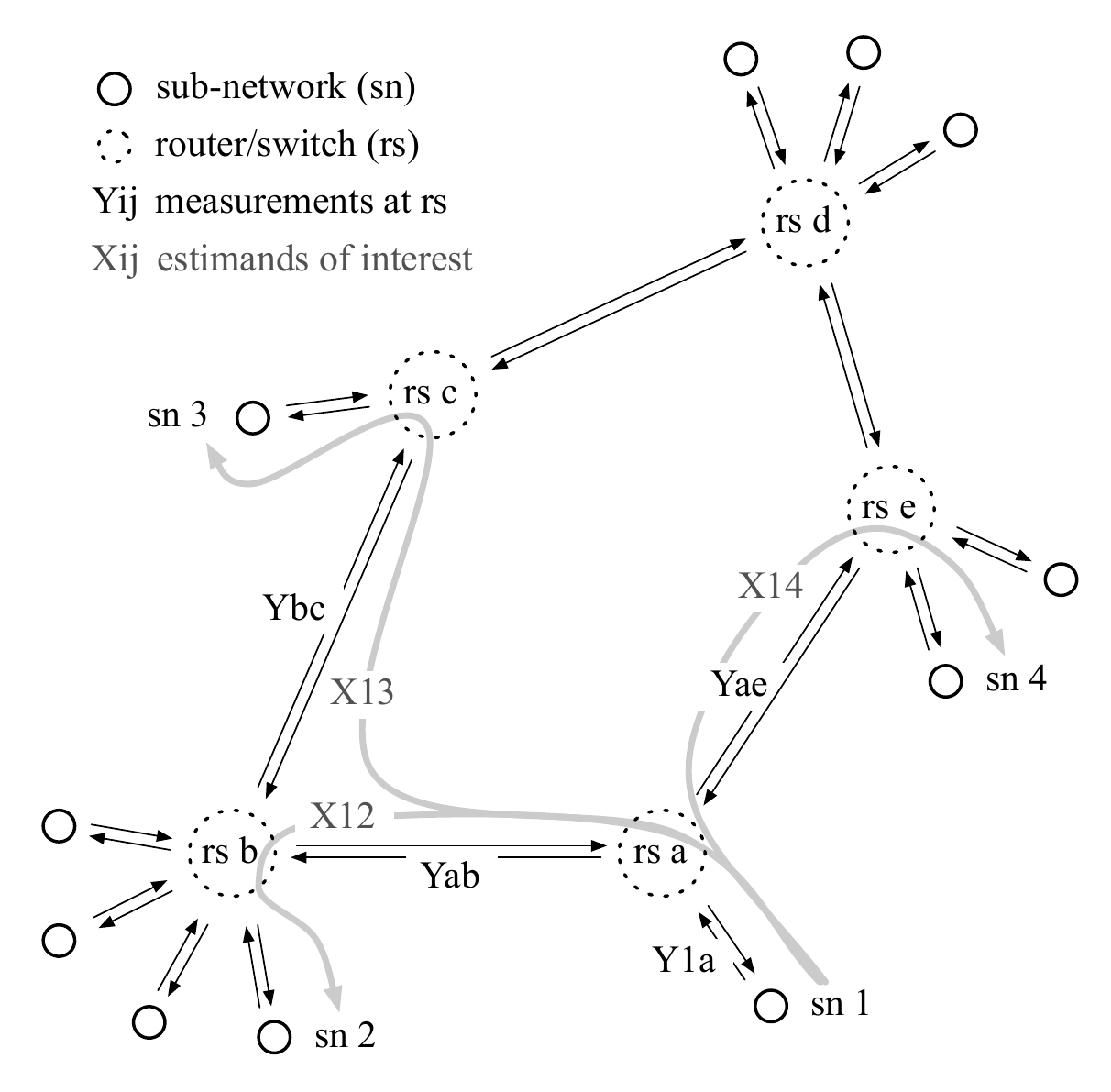}
 \caption{Illustration of the mathematical quantities in network tomography. Traffic $x_{13}$ from sn 1 to sn 3, contributes to counter $Y_{1a}$ into rs $a$, to counter $Y_{ab}$ into rs $b$, to counter $Y_{bc}$ out of rs $b$ (which is the same as counter $Y_{bc}$ into rs $c$), and to counter $Y_{c3}$ out of rs $c$. Traffic volumes on these counters are recorded every few minutes. This routing information is summarized in the column of the routing matrix $A$ that corresponds to the origin-destination traffic volume $x_{13}$.}
\label{fig:nt}
\end{figure}
Black arrows represent represent physical cables connecting routers and switches to subnetworks and indicate the direction in which traffic flows. 
 On each router, a set of (software) counters measure aggregate traffic volumes, $y_{ij}$,
 corresponding to incoming and outgoing cables, routinely (every five minutes).
 The traffic recorded by each of these counters is the sum of a known subset\footnote{This information is encoded by the routing protocol.} of non-observable point-to-point traffic, $x_{ij}$, represented by grey arrows, over the same time window.
 For example, in Figure \ref{fig:nt}, traffic volumes $x_{12},x_{13}, x_{14}$ all contribute to counter $y_{1a}$, and traffic volumes $x_{12},x_{13}$ both contribute to counter $y_{ab}$.
 To establish a formal connection between measurements $y_{ij}$ and estimands $x_{ij}$, it is convenient to simplify notation.
 Let's order all the (from-to) counter measurements collected over a five-minute window, into a vector $\bm y_t \in\mathbb{R}^m$. 
 We have $m=32$ measurements in Figure \ref{fig:nt}. 
 Let's also order\,\footnote{Any such two orderings can be arbitrary and defined independently of each other. Different pairs of orderings will lead to different $A$ matrices.}  all the non-observable point-to-point traffic volumes of interest over a five-minute window, into a single vector $\bm x_t$.
 We have $n=11^2$ point-to-point traffic volumes in Figure \ref{fig:nt}.
Using this more compact notation, we can write  $y_{it} = \sum_{j=1}^n A_{ij} \, x_{jt}$, where $t$ denotes time, and the matrix $A_{m\times n}$ is build using information about the pre-specified routing protocol. In particular, $A_{ij}=1$ if point-to-point traffic $i$ contributes to counter $j$ and $A_{ij}=0$ otherwise, in the case of deterministic routing. More complicated routing schemes, including probabilistic routing and dynamic load-balancing protocols that minimize the expected congestion, can also be formulated in terms of Equation \ref{eq:linearproblem}, as discussed in Section \ref{sec:remarks_extensions}.

From a statistical perspective, the application to communication networks we consider has additional features that make the inference task harder, and more interesting, than in a traditional deconvolution problem.
 First, we have low dimensional observations, $\bm x_t$, and high dimensional estimands, $\bm y_t$.   In Figure \ref{fig:nt}, for example, $m=32$ and $n=121$.
  In a general network with $d$ subnetworks, $m=O(d)$ is often orders of magnitude lower than $n=O(d^2)$, depending on the redundancy of some counters and on whether we are interested on traffic volumes on all possible origin-destination pairs.
 Second, the space where the estimands live is highly constrained. 
  We prove in Section \ref{sec:theory} that the solution space is a convex polytope of dimension $n-m$.
  The dimensionality of this convex polytope gives the true complexity of the problem, in a computational sense.
  Working in a constrained solution space helps the inference to a point (e.g., see Theorem \ref{thm:unimodality}). 
   We gain additional information from modeling traffic dynamics explicitly.
  Sampling from such an extremely constrained solution space, however, proves to be a challenge.
   We approach this sampling problem by combining a random direction sampler \citep{Smith:1984} with model-based regularization and a sequential sample-importance-resample-move (SIRM) particle filter \citep{Gilks:Berzuini:2001}. 

In this paper, we introduce a new dynamic multilevel model for aggregate traffic volumes that posits two latent dynamic processes, in Section \ref{sec:model}.
The first is a heavy-tailed traffic process, in which the amount of traffic on each origin-destination route is proportional to its variability up to a scaling factor shared by all origin-destination routes.
The second is an additional error process for better capturing near-zero traffic volumes.
We carry out inference via a sequential sample-importance-resample-move particle filter, and we develop  a novel two-stage strategy \citep[inspired by][]{CloggRubin1991}, in Section \ref{sec:estimation}.
A transformation of the heavy-tailed layer of the multilevel model can be embedded into a Gaussian state-space formulation with identifiable parameters.
We use the fit for such a reformulation to calibrate informative priors for key parameters of the multilevel model, and to develop an efficient particle filter that is statistically efficient, numerically stable and scales to large problems.
In Section \ref{sec:empirical}, we show that the proposed methods
 are more accurate than published state-of-the-art solutions on two time series data sets.
In Section \ref{sec:simulations}, we then design experiments that combine real and simulated data to investigate comparative performance.
In Section \ref{sec:remarks}, we offer remarks on modeling, inferential and computational challenges with the proposed methods, and discuss limitations and extensions.

The R package \texttt{\small networkTomography} includes the two unpublished data sets we analyze, as well as robust code implementing all the seven methods we compare. It is available on the Comprehensive R Archive Network at~ {\small \url{http://cran.r-project.org/}}.

\subsection{Related work}
\label{sec:litreview}

Applied research related to the type of problems we consider can be traced back to literature on transportation and operations research \citep{bell:1991, vand:iann:1994}. 
There the focus is on estimating a single set of origin-destination traffic volumes, $\bm y$, from integer-valued traffic counts over time, $\bm x_t$.
The line of research in statistics with application to communication networks is due to \citet{vardi:1996} who coined the term \textit{network tomography} by extending the approach to positron emission tomography by \cite{shep:vard:1982}. In this latter setting, statistical approaches may be able to leverage knowledge about a physical process, explicitly specified by a model, to assist the inference task. In the network tomography setting, in contrast, we can only rely on knowledge about the routing matrix and statistics about traffic time series.

From a technical perspective, \citet{vardi:1996} develops an estimating equation framework to estimate a single set of origin-destination traffic volumes from time series data; the same data setting and estimation task considered in the transportation and operations research literature.
\citet{tebaldiwest1998} develop a hierarchical Bayesian model that can be fit at each epoch independently, thus recovering time-varying origin-destination traffic volumes.
They point out that the hardness of the problem lies in having to sample  from a very constrained solution space.
Informative priors are advocated as a means to mitigate issues with non-identifiability and multi-modality that arise when making inference from aggregated traffic volumes at each  point in time.
In previous work \citep{Airo:Falo:2004} we extended their approach by explicitly modeling complex dynamics of the  non-observable time series.
\citet{Cao:Dav:Van:Yu:2000} develop a local likelihood approach to attack the non-identifiability issue.
They develop a Gaussian model with a clever parametrization that leads to identifiability of the  point-to-point traffic volumes, if they are assumed independent over a short time window---approximately 1-hour.
\citet{Cao:Dav:Van:Yu:Zu:2001} extend this approach to inference on large networks by adopting a divide-and-conquer strategy .
\citet{zhang:roug:lund:dono:2003} develop gravity models that can scale to large networks and use them to analyze point-to-point traffic on the AT\&T backbone in North America.
This approach has been extended by \citet{fang:vard:zhan:2007} and \citet{zhang2009}.
Work in this area by \citet{Soule2005} and \citet{Erramilli2006} provide slightly different approaches to this class of problems.
Recent reviews of this literature are available \citep{castro2004,lawr:2006a}.

One of the key technical problems that we face during inference is that of sampling solutions from a convex polytope.
 In this sense, the problem of sampling a feasible set of origin-destination traffic volumes given aggregate traffic is equivalent to that of sampling square tables given row and column totals, when the routing matrix corresponds to a star network topology.
As we consider more complicated topologies, the equivalence still holds for more elaborate specifications of marginal totals.
\citet{airo:haas:2011} characterize such a correspondence using projective geometry and the Hermite normal form decomposition of the routing matrix $A$.
Leveraging this equivalence, the iterative proportional fitting procedure \citep[IPFP, ]{demi:step:1940,fien:1970} provides a baseline for the traffic matrix estimation at each epoch in Section \ref{sec:eval_methods}.
Other approaches to the problem of sampling tables given row and column margins include   a sequential MCMC approach \citep{chendiacholmliu:2005},  a dynamic programming approach that is very efficient for matrices with a low maximum marginal total \citep{harr:2009,mill:harr:2011},  and  sampling strategies based on algebraic geometry \citep{diacstur:1998, dobr:2011} or on an explicit characterization of the solution polytope \citep{airo:haas:2011}.

A related body of work on tomography focuses on the problem of {\em delay (network) tomography}, in which the times traffic reaches/leaves the routers are recorded at the router level, instead of the volumes \citep{pres:duff:horw:tows:2002, lian:yu:2003b, lawr:2006b, deng:li:zhu:liu:2012}. However, inference in delay tomography has a different structure from inference in volume tomography, which  we focus on in this paper.

\section{A model of mixing time series on a network}
\label{sec:model}

Given $m$ observed traffic counters over time, $\bm y_t = \{y_{it} : i=1\dots m\}$, the aggregate traffic loads, we want to make inference on $n$ non-observable point-to-point traffic time series, $\bm x_t = \{x_{jt} : j=1\dots n\}$.
The routing scheme is parametrized by the routing matrix $A$, of size $m\times n$.
Without loss of generality, we consider the case of a fixed routing scheme. In this case, the matrix $A$ has binary entries;  element $A_{ij}$ specifies whether counter $i$ includes the traffic on the point-to-point route $j$.
Extensions to probabilistic routing and dynamic protocols for congestion management are discussed in Section \ref{sec:remarks_extensions}.

The main observation that informs model elicitation is that the measured traffic volumes, $\bm y_{t}$, are  heavy--tailed and sparse.
 For instance, peak traffic may be very high during certain hours of the day, and traffic is often zero during night hours. 
We develop a multilevel state-space model to explain such a variability profile of the observed aggregate traffic volumes.
 The proposed multilevel model involves two latent processes: $\{\bm \lambda_t : t\geq1\}$ at the top of the hierarchy, and $\{\bm x_t: t\geq1\}$ in the middle of the hierarchy.
The observation process is at the bottom of the hierarchy.
Intuitively, we posit a heavy tailed $\{\bm \lambda_t\}$ process, and a thin tailed $\{\bm x_t | \bm \lambda_t \}$ process, specifying $\bm x_t \mid \bm \lambda_t$ as additive error around $\bm \lambda_t$, constrained to be positive.
The key point is that we need both temporal correlation and independent errors to induce positive density for near-zero traffic.
In previous work, we assumed a heavy tailed $\{\bm \lambda_t\}$ process and a heavy tailed $\{\bm x_t\}$ process, specifying $\bm x_t\mid\lambda_t$ as independent log-normal variation conditionally on $\lambda_t$ \citep{Airo:Falo:2004}.
This set of choices, however, leads to some computational instability during inference when actual point-to-point traffic is zero (or nearly zero), as the likelihood for $\bm x_t$ had zero density at $x_{j,t}=0$.

In detail, we posit that each point-to-point traffic volume $x_{j,t}$ has its own time-varying intensity $\lambda_{j,t}$. 
This underlying intensity evolves through time according to a multiplicative process
\[
 \log \lambda_{j,t} = \rho \log \lambda_{j,t-1} + \varepsilon_{j,t}
\]
where $\varepsilon_{j,t} \sim \mathrm{N}(\theta_{1\,j,t}, \theta_{2\,j,t})$. 
Such a process leads to heavy-tailed traffic volumes that are not sparse. 
Moreover, small differences between low traffic volumes receive quite different probabilities under the log-normal model.
Thus, conditional on the underlying intensity, we posit that the latent point-to-point traffic volumes $x_{j,t}$ follow a truncated normal error model,
\[
 x_{j,t} | \lambda_{j,t}, \phi_t \sim
 \mathrm{TruncN}_{(0,\infty)} \left( \lambda_{j,t}, \; \lambda_{j,t}^\tau (\exp(\phi_t) - 1) \right),
\]
where $\tau$ and $\phi_t$ regulate temporally independent variation.
The mean-variance structure of the error model is analogous to that of a log-normal distribution for $\tau=2$; in particular, if $\log(z) \sim \mathrm{N}(\mu, \sigma^2)$, $\mathbb{E}(z) = \exp(\mu + \sigma^2/2)$ and $\Var(z) = \exp(2 \mu + \sigma^2) \cdot (\exp(\sigma^2)-1)$.
Thus, $\lambda_{j,t}$ is analogous to $\exp(\mu + \sigma^2/2)$, and $\phi_t$ is analogous to $\sigma^2$.
The observed aggregate traffic is obtained by mixing point-to-point traffic according to the routing matrix, $\bm{y}_t = A \bm{x}_t$.
The model specification is complete by placing diffuse independent log-Normal priors on $\lambda_{j,0}$. We also place priors on $\phi_{t}$ for stability, assuming $\phi_{t} \sim \mathrm{Gamma}(\alpha, \beta_t / \alpha)$.

This multilevel structure provides a realistic model for the aggregate traffic volumes we measure, which are both heavy-tailed and sparse.
The error model induces sparsity while maintaining analytical tractability of the inference algorithms, detailed in Section \ref{sec:estimation}, by decoupling sparsity control from the bursty dynamic behavior.
The log-Normal layer provides heavy-tailed dynamics and replicates the intense traffic bursts observed empirically, whereas the truncated Normal layer allows for near-zero traffic levels with non-negligible probability.
By combining these two levels, we induce a posterior distribution on point-to-point traffic volumes, the estimands of interest in this problem, which can account for both extreme volumes and sparsity.

In summary, we developed a model for observed $m$-dimensional time series $\{y_t\}$ mixing on a network according to a routing matrix $A$.
The model involves two $n$-dimensional latent processes $\{\lambda_t, x_t\}$, a set of latent variables $\{\phi_t\}$, and constants $\rho,\tau,\alpha,\{\beta_t\}$ and $\{\theta_{1t},\theta_{2t}\}$.
While the parameters $\tau,\rho,$ and $\alpha$ provide some flexibility to the model and can be calibrated through exploratory data analysis on the observed traffic time series, the parameters $\{\theta_{1t},\theta_{2t}, \beta_t\}$ are key to the inference.
Strategies for parameter estimation and posterior inference are discussed in Section \ref{sec:estimation}.

\subsection{Qualitative model checking}
\label{sec:qualitative}

As part of the model validation process, we looked at whether the simulated time series from the model in Section \ref{sec:model} possessed qualitative features of  real time series; namely, sparse traffic localized in time, and heavy tails in distribution of the traffic loads.

We generated a number of time series using parameter values $\tau=2$, $\rho=0.92$, $\theta_{1t}=0$, $\theta_{2t}=\frac{2\,\log(5)}{4}$  for all $t$.
In addition, we set $\phi_t=0.25$, rather than setting the constants $\alpha,\beta_t$ underlying the distribution of $\phi_t$, for simplicity.
 These are realistic values for the constants; they were calibrated on the actual point-to-point traffic volumes from the Bell Labs data set in Section \ref{sec:datasets}.
We used the empirical mean, standard deviation, and autocorrelation of $\{\log x_{it}\}$ for each of these time series combined with the observed level of sparsity to create the results below.

Figure \ref{fig:cdfs} shows the empirical CDF of the two latent processes $\{\bm \lambda_t\}$ and $\{\bm x_t\}$ for one simulated time series.
The $\{\bm \lambda_t\}$ process places more mass in any $\epsilon$ ball around zero relatively to the $\{\bm x_t\}$ process.
The figure confirms our intuition about how the truncated Gaussian error operates.

\begin{figure}[t!]
\centering
\includegraphics[width=0.80\textwidth]{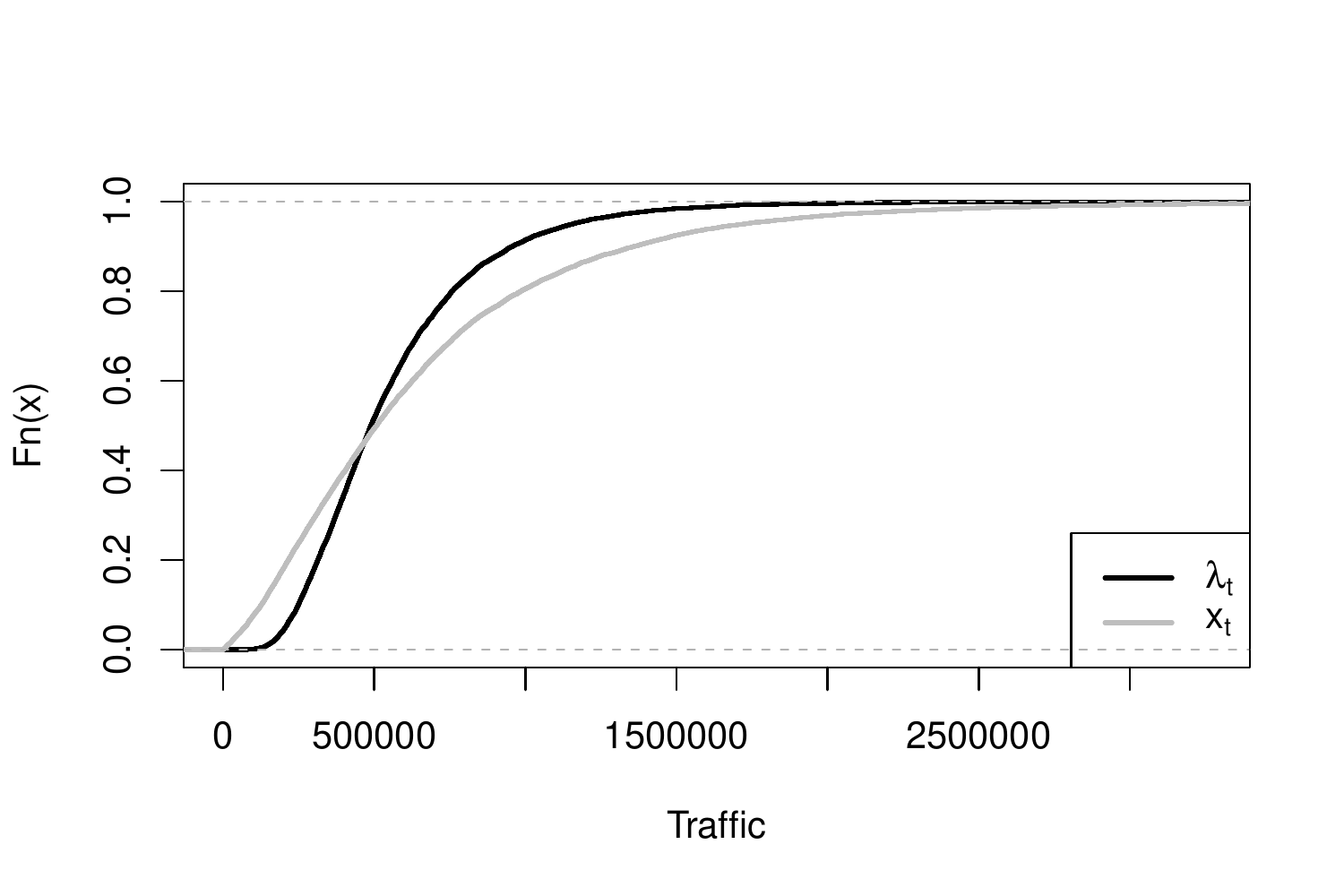}
\caption{Comparison of CDFs for $\lambda_{i,t}$ (solid black line) and $x_{i,t}$ (solid grey line).}
\label{fig:cdfs}
\end{figure}

Figure \ref{fig:simvsreal} shows an origin-destination traffic time series, in the left panel, and a simulated $\{x_t\}$ time series, in the right panel.
Real point-to-point traffic volumes from the router/switch to the local subnetwork were measured using special software installed on the routers, for validation purposes, courtesy of Bell Labs.
The Bell Labs data is further discussed in Section \ref{sec:datasets}.
The simulated time series displays two key qualitative characteristic of the real point-to-point traffic time series.
Specifically,  we observe sudden traffic surges, typical for a heavy tail distribution of traffic volumes, and localized periods of low traffic, as expected from our (truncated) additive Gaussian correction.
\begin{figure}[t!]
\centering
\includegraphics[width=0.80\textwidth]{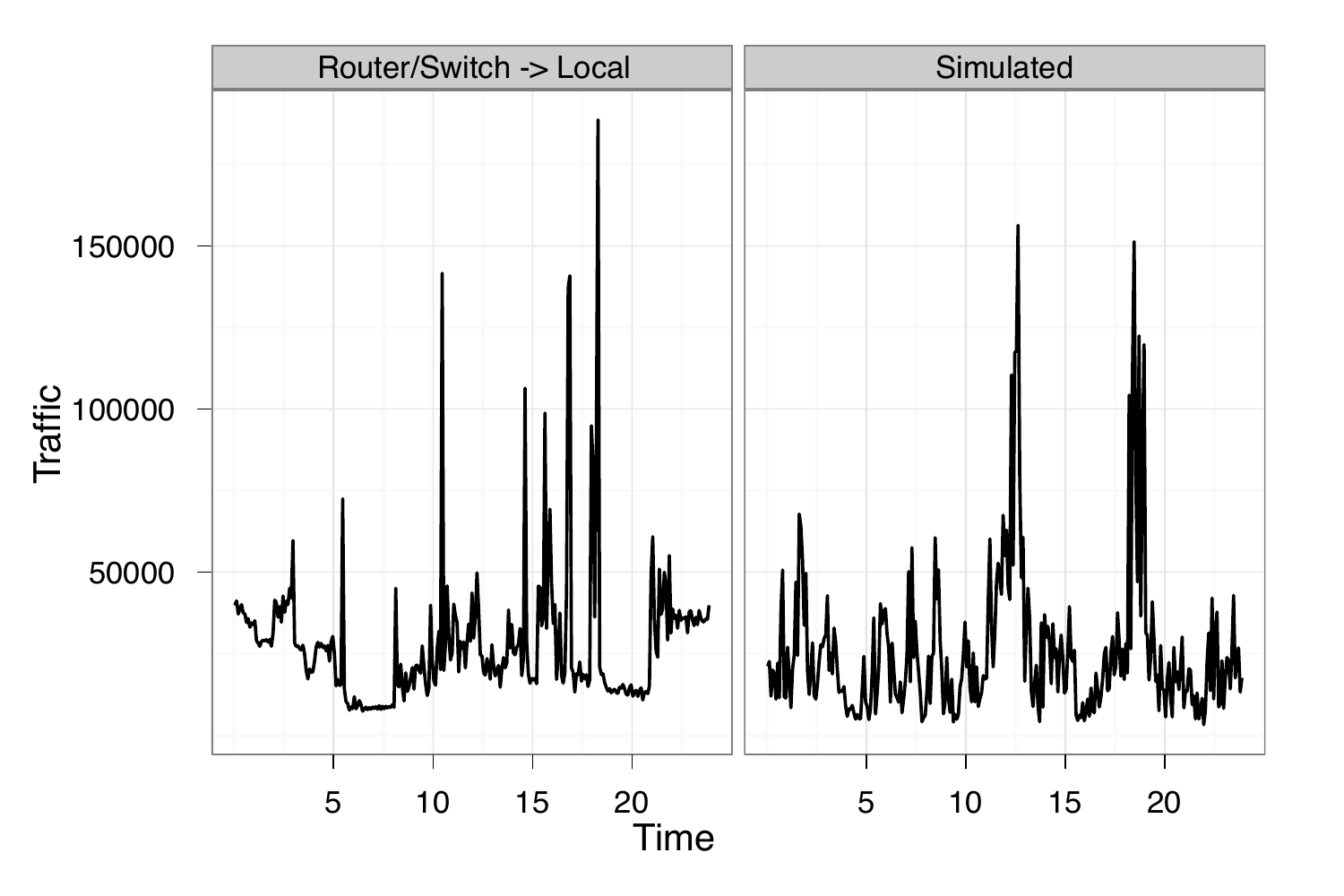}
 \caption{Actual (left panel) versus simulated (right panel)  point-to-point traffic volumes.}
\label{fig:simvsreal}
\end{figure}

The anecdotal findings above hold for most of the real point-to-point traffic volumes and simulated time series we considered.
This suggests that the proposed model is capable of generating data that qualitatively resemble real traffic volumes.
There are two important ways in which observed and simulated traffic differs, though.
First, simulated point-to-point traffic peaks last longer than real traffic peaks.
This is due to the autoregressive dependence in $\lambda_{j,t}$.
Second, simulated point-to-point traffic volumes are more variable than real traffic volumes in low traffic regimes.
This is due to the truncated Normal noise.
Thus, the proposed model is not a perfect generative mechanism for point-to-point traffic.
However, the structure the model captures is sufficient to provide useful posterior inference, as we demonstrate in Section \ref{sec:empirical}.

\subsection{Theory}
\label{sec:theory}

Multi-modality has been reported in the literature as an issue in network tomography. 
This issue has previously been illustrated only in a toy example; a small network with Poisson traffic \citep{vardi:1996}.
We investigated multi-modality from a geometric perspective.
Intuitively, our main result is that whenever a well-behaved distribution is used as a prior for the individual origin-destination traffic volumes, however diffused, the posterior cannot have disconnected modes.
The main result applies  directly to the case of real-valued traffic volumes \citep[e.g.,][and this paper]{Cao:Dav:Van:Yu:2000,Airo:Falo:2004}.
For the case of integer-valued traffic volumes, analyzed by others \citep[e.g.,][]{vardi:1996,tebaldiwest1998}, only a weaker condition is possible.

Consider the case of real-valued non-negative traffic volumes.
Feasible traffic volumes $\bm x_{t}$ must be non-negative and satisfy $y_{it} \geq \sum_{j} A_{ij}\, x_{jt} $.
In other words, the space where $\bm x_t$ lives can be characterized as the intersection between the positive orthant and $m$ half-spaces in $n-m$ dimensions.
This is a convex polyhedron.
Since both $\bm y_t$ and $A_{ij}$ are non-negative, the polyhedron is bounded and the space of feasible solution is a convex polytope.
The main result is a consequence of the fact that the space in which solution vectors $\bm x_t$ to Equation \ref{eq:linearproblem} is a convex polytope.

\begin{theorem}
\label{thm:unimodality}
Assume $f(\bm x_t)$ is quasi-concave. Let $\bm y_t = A \bm x_t$. Then, $f(\bm x_t | \bm y_t)$ will also be quasi-concave, and will have no separated modes. The set $\left\{ \bm z : f(\bm z | \bm y_t) = \max_{\bm w} f(\bm w | \bm y_t) \right\}$ is connected.
\end{theorem}
\begin{proof}
$f(\bm x_t | \bm y_t) \propto I( y_t = A \bm x_t ) f(\bm x_t)$, so $f(\bm x_t | \bm y_t)$ has support on only a bounded $n-m$ dimensional subspace of $\mathbb{R}^n$, which forms a closed, bounded, convex polytope in the positive orthant. Denote this region $B(\bm y_t)$.
Denote the mode of $f(\bm x_t)$ as $\hat{\bm x}_t$. We now consider two cases.

\noindent\textit{Case 1.} $\hat{\bm x}_t \in B(\bm y_t)$. Then, the mode of $f(\bm x_t | \bm y_t)$ is also $\hat{\bm x}_t$.

\noindent\textit{Case 2.} $\hat{\bm x}_t \notin B(\bm y_t)$. Then, we must do a little more work.
Consider the level surfaces of $f(\bm x_t | \bm y_t)$, denoting $C(z) = \left\{ \bm u : f(\bm u | \bm y_t) = z \right\}$.
Define $z^* = \max_{B(\bm y_t)} f(\bm x_t | \bm y_t)$; this is well-defined and attained as $B(\bm y_t)$ is closed.
Now, denoting $C_0(z) = \left\{ \bm u : f(\bm u) = z \right\}$, we have $C(z) = C_0(z) \bigcap B(\bm y_t)$. As $f(\bm x_t)$ is quasiconcave, its superlevel sets $U_0(z) = \left\{ \bm u : f(\bm u) \geq z \right\}$ are convex.
Thus, the superlevel sets of $f(\bm x_t | \bm y_t)$, denoted $U(z) = U_0(z) \bigcap B(\bm y_t)$ analogously, are also convex.
So, we have that the set $U(z^*) = C(z^*)$ is convex and non-empty.
Therefore, we have established that set of modes for $f(\bm x_t | \bm y_t)$ is convex, hence connected.
\end{proof}

Next, consider the case of integer-valued non-negative traffic volumes.
To precisely state conditions under which pathological behavior is not possible in the integer-valued case we need to introduce some concepts from integral geometry.
A square integer matrix is defined to be \emph{unimodular} if it is invertible as an integer matrix (so its determinant is $ \pm 1 $).
By extension, we define a rectangular matrix to be \emph{unimodular} if it has full rank and each square submatrix of maximal size is either unimodular or singular (its determinant is $0$). 

With integer-valued traffic, the inferential goal is to sample solutions to Equation \ref{eq:linearproblem}, where $ A $ is a given unimodular $ m \times n $ matrix with $ \{0,1\} $ entries and $ \bm y_t $ is a given integer positive vector.
In the case of real-valued traffic, it was straightforward to show that the space of solutions to \ref{eq:linearproblem} is a convex polytope. %
In the case of integer-valued traffic we have:
\begin{theorem}[Airoldi \& Haas, 2011]
\label{lem:sol-int-pol}
The space of real solutions $x$ to equation $y=A\,x, ~x \geq 0$, is an integral polytope, whenever $A$ is unimodular.
\end{theorem}
\begin{proof}
The vertices are the intersections of the affine solution space of $ Ax = y $ with the $ (n-m) $-coordinate planes bordering the non-negative orthant.
So a vertex $ x $ has $ n-m $ zero coordinates.
Let's gather the rest of the coordinates into a positive integer vector $ x' $ of dimension $ m $.
And let's gather the corresponding columns of $ A $ into a square matrix $ A_1 $; so we get the equation $ A_1x'=y $.
If $ A_1 $ was singular, the latter system would have either none or infinitely many solutions, which would contradict that $ x $ is a vertex.
So $ A_1 $ is unimodular and $ x' = A_1^{-1}y $.
And since $ y $ is integer, $ x' $ is also integer, and so is $ x $.
\end{proof}
We can precisely characterize the space of feasible traffic volumes in the integer case, however, we cannot directly address multi-modality. The concept of multiple modes and local maxima are not well-defined in this setting. This results, however, provides insight into the connection between our results and the pathological case demonstrated by \citet{vardi:1996}.

The theory above helps us settle an important question about our model: how will posterior inference behave under dynamic updates? If dynamic updates were allowed to ``grow'' modes over time or exhibit other pathological behavior, the computation would be quite difficult and  inference results would be less credible. Fortunately, this is not the case. In general, we have established that the quasiconcavity of a predictive distribution $f(\bm x_t | \bm y_{t-1}, \ldots)$ implies the quasiconcavity of the posterior $f(\bm x_t | \bm y_t, \bm y_{t-1}, \ldots)$; thus, the set of maxima for $f(\bm x_t | \bm y_t, \bm y_{t-1}, \ldots)$ will form a convex set under the given condition. Since we initialize our model with a unimodal (quasiconcave) log-Normal distribution and impose log-Normal dynamics on the underlying intensities $\bm \lambda_t$, Theorem  \ref{thm:unimodality} provides a useful limit on  pathological behavior during inference with our model.

The situation is somewhat similar, but less constrained, in the case of integer traffic volumes, for unimodular routing matrices. While it is not known under what conditions a network routing scheme translates into a unimodular routing matrix $A$, the routing matrices in the cited literature are all unimodular. Thus, extreme forms of multi-modality can be ruled out from the literature on dynamic network tomography in many cases.
Our theory also suggests that models based upon real-valued traffic volumes will exhibit more predictable behavior under posterior updates than those based upon integer-valued volumes, making the former much more attractive for inference in cases where integer constraints provide little addditional information.

\section{Parameter estimation and posterior inference}
\label{sec:estimation}

Here we develop two inference strategies to produce estimates for the point-to-point traffic time series of interest, using the model described in the previous section.
The first strategy is based on a variant of the sequential sample-importance-resample-move (SIRM) particle filter \citep{Gilks:Berzuini:2001}. 
This filter is simple to state and implement, but is computationally expensive due to the large number of particles needed to explore probable trajectories in high-dimensional polytopes. 
Details are given in Section \ref{sec:sirm}.
The second strategy combines the sequential SIRM filter with a model-based regularization step that leads to efficient particles.
This strategy preferentially explores trajectories in regions of the solution polytopes with high posterior density.
The model-based regularization step involves fitting a Gaussian state-space model with an identifiable parametrization, which leads to informative priors for the multilevel  model that are leveraged by a modified SIRM filter.
Details are given in Section \ref{sec:2stage}.

\subsection{A SIRM filter for multilevel state-space inference}
\label{sec:sirm}

Inference in the multilevel state-space model is performed with a sequential sample-resample-move algorithm, akin to \citet{Gilks:Berzuini:2001}. Its structure is outlined in Algorithm \ref{alg:sirm}.

\begin{algorithm}[ht!]
\hspace{-8pt} \textbf{Sample-Importance-Resample-Move algorithm}\\ %
 \For{$t \leftarrow 1$ \KwTo $T$} {
 \hspace{-8pt} Sample step: \\
 \For{$j \leftarrow 1$ \KwTo $m$} {
  Draw a proposal $\log \lambda_{i,t}^{(j)*} \sim \mathrm{N}(\theta_{1\,i,t} + \log \lambda_{i, t-1}^{(j)}, \theta_{2\,i,t}) $ \\
  Draw $\phi_t^{(j)} \sim \mathrm{Gamma}(\alpha, \beta_t/\alpha)$ \\
  Draw $\bm x_t^{(j)*}$ from a truncated Normal distribution with mean $\bm \mu^{*} = \rho/m \sum_{j=1}^{m} \bm \lambda_{t-1}^{(j)}$ and covariance matrix $\Sigma^{*} = (\exp(\beta_t) - 1) \diag{ \bm \mu^{*2} }$ on the feasible region given by $\bm x_t^{(j)*} \ge 0$, $y_t = A x_t^{(j)*}$ using Algorithm \ref{alg:rda} \\
 }
 Resample our particles $(\bm \lambda_t^{(j)*}, \phi_t^{(j)*}, \bm x_t^{(j)*})$ with probabilities proportional to our weights $w_t^{(j)}$ \\
 Move each of our resampled particles $(\bm \lambda_t^{(j)}, \phi_t^{(j)}, \bm x_t^{(j)})$ using a MCMC algorithm (Metropolis-Hastings within Gibbs, with proposal on $\bm x_t$ given by Algorithm \ref{alg:rda})
 }
 \textbf{return} $(\lambda_t^{(j)}, \phi_t^{(j)}, \bm x_t^{(j)})$ for $j \leftarrow 1$ \KwTo $m$, $t \leftarrow 1$ \KwTo $T$ \\ %
 \medskip
 \caption{SIRM algorithm for inference with multilevel state-space model \label{alg:sirm}}
\end{algorithm}

In Algorithm \ref{alg:sirm}  we use  a random-walk prior on the latent intensities $\bm \lambda_t$. Thus, we fix $\theta_{1\,i,t}=0$ for all $i,t$, and calibrate the constants $\{\theta_{2\,i,t}\}$, %
$\{\beta_t\}$, %
$\alpha, \tau$ and $\rho$ as discussed in Section \ref{sec:constants}.

Sampling particles that correspond to feasible trajectories, that is, to point-to-point traffic volumes $\bm x_t$ in the convex polytope implied by $\bm y_t = A \, \bm x_t$, is non-trivial. 
The use of a random direction proposal on the region of feasible point-to-point traffic is a vital component of the SIRM filter. %

We use the {\it random directions algorithm} \citep[RDA, ][]{Smith:1984} to sample from the distributions of feasible traffic volumes on a constrained region, in the SIRM filter.
This method constructs a random-walk proposal on a convex region, such as the feasible regions for $\bm{x}_t$, by first drawing a vector $\bm d$ uniformly on the unit sphere.
It then calculates the intersections of a line along this vector with the surface of the bounding region, and samples uniformly along the feasible segment of this line.
Computing the feasible segment is facilitated by decomposing $A$. %
We decompose $A$ as $[A_1 \mid A_2]$ by permuting the columns of $A$, and the corresponding components of $\bm x_t$, so that $A_1$ $(r \times r)$ is of full rank.
Then, splitting the permuted vector $\bm x_t = [\bm x^{1}_{t} , \bm x^{2}_{t}]$, we obtain $\bm x^1_{t} = A_1^{-1} (\bm y_t - A_2 \bm x^2_{t})$.
This formulation can be used to construct an efficient random directions algorithm to propose feasible values of $\bm x_t$.
We have included pseudocode for this algorithm in Algorithm \ref{alg:rda}.
\begin{algorithm}[t!]
 \hspace{-8pt}\textbf{Random Directions Algorithm} \\ %
\smallskip
  \hspace{-4pt}\textbf{Initialization} \\
  \Begin{
   Decompose $A$ into $[A_1 \, A_2]$, $A_1$ $(r \times r)$ full-rank \\
   Store $B := A_1^{-1}$; $C := A_1^{-1} A_2$ \\
  }
  \hspace{-4pt}\textbf{Metropolis step} \\
  \textbf{given} $\bm x_t$ \\
  \Begin{
   Draw $\bm z \sim N(0, I)$, $\bm z \in \mathbb{R}^{c-r}$ \\
   Set $\bm d := \bm z / \| \bm z \|$ \\
   Calculate $\bm w := C \cdot d$ \\
   Set $h_1 := \max \{ \min_{k: w_k > 0} (\bm x_{1,t})_k / w_k, 0 \}$ \\
   Set $h_2 := \max \{ \min_{k: d_k < 0} -(\bm x_{2,t})_k / d_k, 0 \}$ \\
   Set $h := \min \{ h_1, h_2 \}$ \\
   Set $l_1 := \max \{ \max_{k: w_k < 0} (\bm x_{1,t})_k / w_k, 0 \}$ \\
   Set $l_2 := \max \{ \max_{k: d_k > 0} -(\bm x_{2,t})_k / d_k, 0 \}$ \\
   Set $l := \max \{l_1, l_2\}$ \\
   Draw $u \sim \Unif(l, h)$ \\
   Set $\bm x_{2,t}^* := \bm x_{2,t} + u \cdot \bm d$; $\bm x_{1,t}^* = \bm x_{1,t} - u \cdot w$; $\bm x_t^* = (\bm x_{1,t}^*, \bm x_{2,t}^*)$ \\
   Set $\bm x_t := \bm x_t^*$ with probability $\min \{ f(\bm x_t^*)/f(\bm x_t), 1 \}$
  }
  \textbf{return} $\bm x_t$ \\ %
\medskip
\caption{RDA algorithm for sampling from $f(\bm x_t)$, truncated to the feasible region given by $A \cdot \bm x_t = \bm y_t$ \label{alg:rda}}
\end{algorithm}

All draws from this proposal have positive posterior density, since they are feasible.
This property allows our sampler to move away from problematic boundary regions of the extremely constrained solution polytope.
In contrast, methods that use Gaussian random-walk proposal rules, for instance, can perform quite poorly in these situations, requiring an extremely large number of draws to obtain feasible proposals.
For example, with $\bm x_t \in \mathbb{R}^{16}$, it can sometimes require on the order of $10^9$ draws to obtain a feasible particle, when using the conditional posterior from $t-1$ as proposal.
This is a situation we encountered with alternative estimation methods described in Section \ref{sec:empirical}.

\subsubsection{Setting the constants}
\label{sec:constants}

To carry out inference, we must set values for the constants underlying the distributions at the top layer of the multilevel model; $\{ \theta_{1\,j,t} \}$, $\{ \theta_{2\,j,t} \}$, $\{ \beta_t \}$, $\alpha, \rho$, and $\tau$. 
Choices can be evaluated using small sets of point-to-point traffic  collected for diagnostic purposes, as in Section \ref{sec:qualitative}

The (fixed) autocorrelation parameter $\rho$ drives the dynamics of $\log \lambda_{j,t}$. We typically set $\rho = 0.9$. 
A high value for $\rho$ is a practically plausible assumption, as point-to-point traffic volumes tend to be highly autocorrelated in communication networks \citep{cao:clev:lin:sun:2002}. 

The parameter $\tau$ controls the skew of point-to-point traffic volumes. 
The distribution of point-to-point traffic has been found to be extremely skewed empirically \citep{Airoldi:2003}, and this skew is comparable to the skew of the aggregate traffic volumes.
\citet{Cao:Dav:Van:Yu:2000} found that the local variability of the aggregate traffic volumes is well-described by $\tau=2$. 
In our analyses, we fix $\tau=2$. This assumption was checked on pilot data as in \citet{Cao:Dav:Van:Yu:2000}.

The inference strategy based on the SIRM filter is amenable to a wide range of techniques for regularization.
The simplest of these is a random walk prior on $\log \bm \lambda_t$. 
For this, we fix $\theta_{1\,i,t}=0$ for all $t$ and set $\{\theta_{2\,i,t}\}$ by looking at the observed variability of $\{ \bm y_t \}$.
On the data sets we consider, $\theta_{2\,i,t}=\frac{2~\log 5}{4}$ appeared reasonable based on the variability of the observed aggregate traffic.
That is, we set $\theta_2$ by rescaling the average variance of $\log y_{j,t} - \rho \log y_{j,t-1}$ to correct for aggregation.
This is a somewhat crude approach, but it provides a reasonable starting point.

The collection of constants $\{\beta_t\}$ controls the common scale of variation in the point-to-point traffic.
These constants were  set by examining the observed marginal distribution of $\{ \bm y_t \}$.
We selected $\beta_t = 1.5$ as reasonable value based on the observed excess abundance of values near zero.
Last, the constant $\alpha$ is a fixed tuning parameter. We set it to $n/2$ to provide a moderate degree of regularization for our inference, providing a weight equivalent to $1/2$ of the observed data.

The random walk prior is a simple starting point for our inference and provides cues of computational issues.
Its use is not recommend in practice.
The inverse problem we confront in network tomography is too ill-posed for such a simplistic approach to regularization.
A more refined, adaptive strategy is necessary to provide useful answers in realistic settings.

\subsection{Two-stage inference}
\label{sec:2stage}

Here, we develop an inference strategy that improves the SIRM filter in Algorithm \ref{alg:sirm} by adding a regularization step that guides our inference, focusing our particle filter and sharing information across multiple classes of models.
The idea is to leverage a first-stage estimation step to calibrate informative priors for key parameters in the multilevel model, in the spirit of empirical Bayes \citep{CloggRubin1991}.
Different forms of model-based regularization are feasible (and useful) depending upon traffic dynamics and the topology of a given network.
One approach is to use simple, well-established methods such as gravity-based methods \citep{zhang:roug:duff:gree:2003,zhang:roug:lund:dono:2003,fang:vard:zhan:2007}.
Another approach, developed below, uses a specific parametrization of a Gaussian state-space model to approximate Poisson traffic.
We find that these two approaches are useful in different situations (namely, local area networks and internet backbone networks) as we discuss in Section \ref{sec:empirical}.

\subsubsection{Model-based regularization}
\label{sec:regularization}

Here we describe a simple model used to calibrate key regularization parameters $\{\theta_{1t},\theta_{2t}, \beta_t\}$ of the multi-level state-space model. We posit that $\bm x_t$ follows a Gaussian autoregressive process,
\begin{equation}
 \label{eq:lds+}
  \left \{ \begin{array}{rcl}
    \bm{x}_t & = & F \cdot \bm{x}_{t-1} + Q \cdot \bm{1} + \bm{e}_t \\
    \bm{y}_t & = & A \cdot \bm{x}_t +  \bm{\epsilon}_t.
  \end{array} \right.
\end{equation}
This model can be subsumed into a standard Gaussian state-space formulation, as detailed in Eq.~\ref{eq:gssm}.
\begin{eqnarray}
  && 
  =\left \{ \begin{array}{rcl}
    \left[ \bv{c} \bm{x}_t\\ \bm{1}\ev \right] & = & \left[ \bv{cc} F& Q\\ 0& I\ev \right]  \left[ \bv{c} \bm{x}_{t-1}\\  \bm{1} \ev \right] +  \left[ \bv{c} \bm{e}_t\\ \bm{0}\ev \right] \\
     \bm{y}_t & = & \left[ A \right|  \bm{0} \left] \right.  \left[ \bv{c} \bm{x}_t\\ \bm{1}\ev \right] + \bm{\epsilon}_t
  \end{array} \right. \nonumber  \vspace{3pt} \\ 
  &&
  \label{eq:gssm}
 =\left \{ \begin{array}{rcl}
    \tilde{\bm{x}}_t  & = & \tilde{F} \cdot \tilde{\bm{x}}_{t-1}  + \tilde{\bm{e}}_t  \\
    \bm{y}_t  & = & \tilde{A}  \cdot \tilde{\bm{x}}_t  + \bm{\epsilon}_t.
  \end{array} \right. 
\end{eqnarray}

We estimate $\bm Q$ and $\Cov \bm e_t$, fixing the remaining parameters.
$F$ is fixed at $\rho \bm I$ for simplicity of estimation, with $0.1$ a typical value for $\rho$.
We also fix $\Cov \bm{\epsilon}_t$ at $\sigma^2 \bm I$, with $0.01$ a typical value for $\sigma^2$.
We assume $\bm Q$ to be a positive, diagonal matrix, $Q = \diag{\bm \lambda_t}$, and specify $\Cov \bm e_t$ as $\bm \Sigma_t = \phi \, \diag{\bm \lambda_t}^\tau$, where the power is taken entry-wise.
We obtain inferences from this model via maximum likelihood on overlapping windows of a fixed length.
We develop an inference strategy for this model in Section \ref{sec:empirical}, and provide computational details in Appendix \ref{supp:ssm}.

The model in Eq. \ref{eq:lds+} contains the local likelihood model of \citet{Cao:Dav:Van:Yu:2000} as a special case, when $\rho=0$. 
The marginal likelihood for this model depends only upon the means and covariances of the data. 
A desirable property of this model is that its parameters are identifiable, under conditions analogous to those given in \citet{Cao:Dav:Van:Yu:2000}, for a fixed value of $\rho$.

\subsubsection{Identifiability}

Identifiability in network tomography is a delicate and complex issue \citep{Singhal2007}.
For the proposed model, however, it suffices to consider the marginal distribution of the $\bm y_t$'s.
Under the conditions on the routing matrix $A$ analogous to those in \citet{Cao:Dav:Van:Yu:2000}, the marginal mean and covariance of $\bm y_t$ is an invertible function of the parameters $\bm \lambda$ and $\phi$.
This argument is straightforward with the steady-state initialization discussed in Appendix \ref{supp:ssm}, but it extends to more general settings.
In the case of steady-state initialization, the following result holds.

\begin{theorem}
Assume $\bm y_1, \ldots, \bm y_T$ are distributed according the model given by Eq. \ref{eq:gssm} and described above.
Further assume that $\left| \rho \right| < 1$ and the model is initialized from its steady state---that is, $\bm x_0 \sim N\left( \frac{1}{1-\rho} \bm \lambda, \frac{\phi}{1-\rho^2} D \right)$, where $D = \diag{\bm \lambda_t}^\tau$.
Then, $(\bm \lambda, \phi, \rho)$ is identifiable under the same conditions required for the identifiability of the locally IID model of \cite{Cao:Dav:Van:Yu:2000}.
\end{theorem}
\begin{proof}
The observations $(\bm y_1, \ldots, \bm y_T)$ are jointly normally distributed under the given model.
Further, assuming $\left| \rho \right| < 1$ and steady-state initialization, $\bm y_t \sim N\left( \frac{1}{1-\rho} A \bm \lambda, \frac{\phi}{1-\rho^2} A D A^\top + \sigma^2 I \right)$ marginally for $t = 1, \ldots, T$.
Define $B$ as the matrix containing the rows of $A$ and all distinct pair-wise component-wise products of $A$'s rows.
Fixing $\rho$, these marginal moments are invertible functions of $(\bm \lambda, \phi)$ if and only if the matrix $B$ has full column rank by Theorem 1 of \cite{Cao:Dav:Van:Yu:2000}.
Further, as $\Cov(\bm y_t, \bm y_{t+k}) = \frac{\phi \rho^{|k|}}{1-\rho^2} A D A^\top$, $\rho$ is also identifiable from the component-wise autocorrelations of $\bm y_t$.
\end{proof}

A sufficient condition for $B$ to have full column rank is for A to include aggregate incoming and outgoing (source and destination) traffic for each node, as discussed in \cite{Cao:Dav:Van:Yu:2000}.
This condition holds for all examples we consider and can be checked in practice with pilot studies; such aggregate traffic volumes are easily obtained via network management protocols such as SNMP and are standard input for the widely-used gravity method.
Less restrictive conditions are possible based on the results of \citet{Singhal2007}; however, they are not needed for the situations we consider.

Next, we describe how this model is used to calibrate priors for $\{\lambda_t\}$ and $\{\phi_t\}$ in the multi-level state-space model.

\subsubsection{Calibrating key regularization parameters}
\label{sec:calibration}

To calibrate priors for $\{\lambda_t\}$ and $\{\phi_t\}$ in the multilevel state-space model, we follow a few steps. 
 First obtain estimates from the Gaussian state-space model in the previous Section.
 We correct the estimates at each epoch through the iterative proportional fitting procedure (IPFP) to ensure positivity and validity with respect to our linear constraints. 
 We then smooth the corrected estimates using a running median with a small window size (consisting of 5 observations) to obtain a final set of $\bm{\hat x}_t$ estimates. 
This smoothing step is important as it removes outlying estimates, which often originate from computational errors, from the prior calibration procedure.
These outliers can otherwise degrade the effectiveness of the regularization. We have observed some sensitivity to the choice of window sizes---too broad and it smooths out bursts of traffic, too narrow and outlying estimates compromise our regularization.
We selected 5 as the narrowest window that empirically removed outliers; we recommend this as a guideline for other settings.
These final $\bm{\hat x}_t$ estimates are used to set the mean traffic intensity for $\lambda_t$ as follows,
\[
 \theta_{1\,j,t} = \log \hat x_{j,t} - \rho \log \hat x_{j,t-1}.
\]
The variability of the traffic intensity $\theta_{2\,j,t}$ is set using the estimated variance of the final estimates $\hat x_{j,t}$. Denoting the estimated final variances with $\hat V_{j,t}$, we set $\theta_{2\,j,t}$ as follows,
\[
 \theta_{2\,j,t} = (1-\rho^2) \, \log(1 + \hat V_{j,t} / \hat x_{j,t}^2).
\]

The estimated $\{ \hat{\phi}_t \}$ in the Gaussian state-space model are used to calibrate the prior for the corresponding parameter $\{\phi_t\}$ in the multilevel state-space model. In particular, we set $\beta_t = \log(1 + \hat{\phi}_t)$. The form of this calibration is based on the log-Normal variance relationship described in Section \ref{sec:model}. Remaining constants are calibrated as described in Section \ref{sec:constants}.

Alternative calibration approaches are possible, which use estimates from simple models to calibrate regularization parameters.
For instance, in Section \ref{sec:empirical} we consider a simple gravity model in addition to the state-space model described above.
We take each gravity estimate to be $\hat x_{j,t}$, and we set each $\theta_{1\,j,t}$ as above.
With simpler model we recommend using an empirical approach to setting $\theta_2$; using the gravity model estimates, we set each $\theta_{2\,j,t}$ equal to the overall variance of $\theta_{1\,j,\cdot}$.

\section{Empirical analysis of traffic data}
\label{sec:empirical}

Here, we present the analysis of three aggregate traffic data sets, for which origin-destination (OD) traffic volumes were also collected with special software over a short time period. 
The first data set involves traffic volumes on a local area network with 4 nodes (16 OD pairs)  at Bell Labs, previously analyzed in \citet{Cao:Dav:Van:Yu:2000}.
The second data set involves traffic volumes on a local area network with 12 nodes (144 OD pairs)  at Carnegie Mellon University, previously analyzed by \citet{Airo:Falo:2004}.
The final data set consists of traffic volumes from the Abilene network, an Internet2 backbone network with 12 nodes (144 OD pairs) previously analyzed in \cite{fang:vard:zhan:2007}.
We use these three data sets to evaluate the proposed deconvolution methods. 
We compare the performance of our approach to that of several previously presented in the literature for this problem, focusing on accuracy, computational stability, and scalability.

We find that, of the seven methods we compare, the proposed methods consistently outperform all others both in terms of  $L_1$ and $L_2$ estimation error. The empirical evaluation we provide below uses communication networks that are among the largest ever tried on this problem in the statistics and computer science literature. A quantitative evaluation is possible, since ground-truth origin-destination traffic was laboriously collected for the three network we consider. 

An R package that includes these three data sets and code to replicate the analyses below is available on CRAN, in the \texttt{\small networkTomography} package.

\subsection{Data sets}
\label{sec:datasets}

The first data set was provided courtesy of Jin Cao of Bell Labs. 
We analyze the  traffic volumes measured at {\em router1}, with four subnetworks organized as in Figure \ref{fig:networks} (left panel).
These yield eight observed aggregate traffic volumes (seven of them are independent, since the router does not send, nor receives traffic) and 16 origin-destination traffic volumes \citep{Cao:Dav:Van:Yu:2000}. 
The aggregate traffic volumes are measured every five minutes over one day on the Bell Labs corporate network. 
This yields multivariate measurements at 287 points in time.
The small size of this network allows us to focus on the fundamentals of the problem, avoiding scalability issues.

The second data set was collected at the Information Networking Institute of Carnegie Mellon University, courtesy of Russel Yount and Frank Kietzke.
For the purpose of this paper, given that the topology of Carnegie Mellon network is sensitive, we built a data set of aggregate traffic volumes by mixing two days of origin-destination traffic volumes on a slightly modified network topology.
The network topology we use consists of 12 subnetworks, organized as in Figure \ref{fig:networks} (right panel).
These are connected by two routers, one with four of the nodes, the other with the remaining eight nodes.
The routers are linked via a single connection.
This configuration yields 26 observed aggregate traffic volumes and 144 origin-destination traffic volumes, observed every five minutes at 473 points in time.
\begin{figure}[b!]
 \centering
\includegraphics[width=0.75\textwidth]{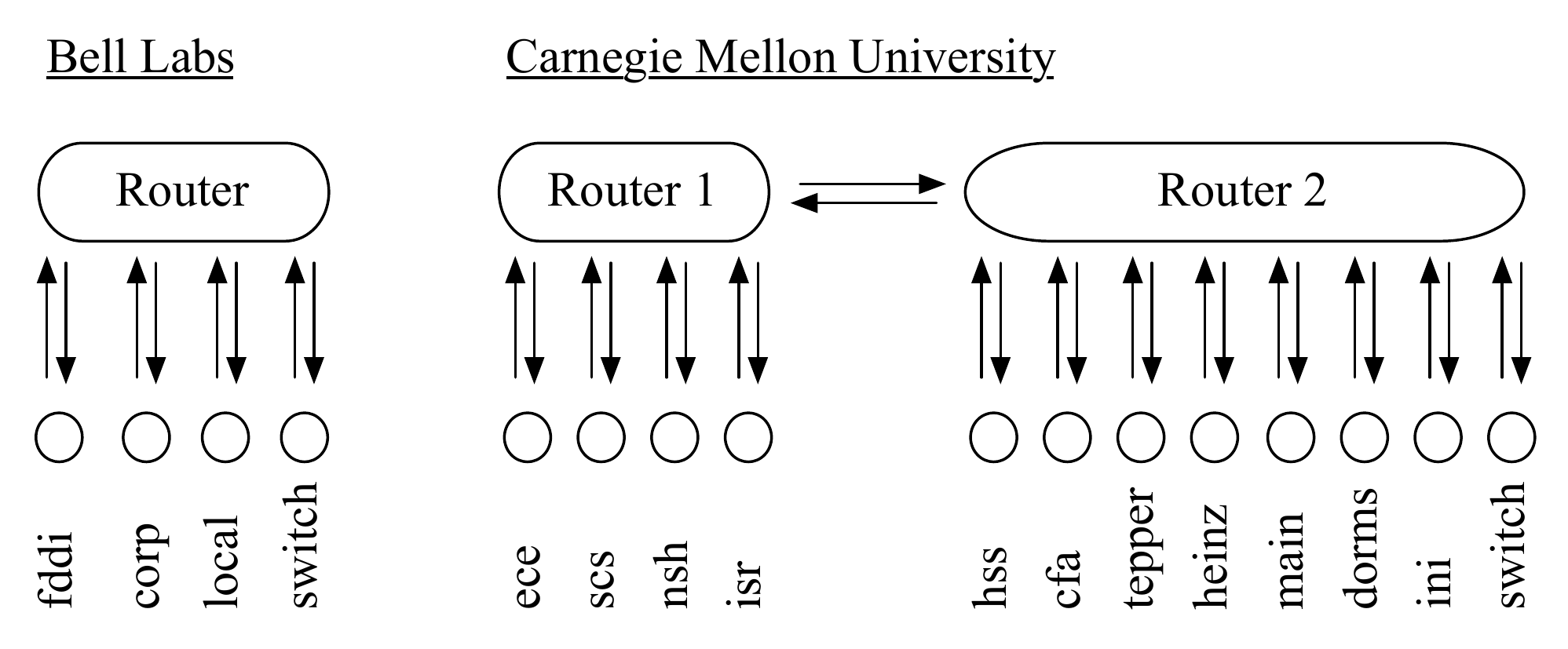}
\caption{Topologies of the Bell Labs and Carnegie Mellon networks. The Abilene network has a more complex topology, an illustration of which is available in \cite{fang:vard:zhan:2007}. We also include these data sets in the \texttt{networkTomography} package.}
\label{fig:networks}
\end{figure}
This larger data set allows us to compare network tomography techniques in a richer, more realistic setting.
In combination with the {\em router1} data, it also allows us to explore the effect of dimensionality on performance and computational efficiency on real traffic data.
Neither data set contained any missing observations.

Our third data set comes from the Abilene network, courtesy of Matthew Roughan.
We use the {\it X1} dataset analyzed in \cite{fang:vard:zhan:2007}, which consists of  aggregate and point-to-point traffic volumes measured every five minutes over a 3-day period. %
The underlying network has 12 nodes, yielding 144 point-to-point traffic time series.
A total of 54 aggregate traffic volumes are observed at each time, consisting of 30 inner links and 24 edge links.
Abilene is an Internet2 backbone network. Abilene's traffic volumes, dynamics and variability are quite different from those observed on local area networks such as Carnegie Mellon's and Bell Labs'. Abilene's topology is more complicated than simple star and dual-loop configurations.
Thus this data set provides a different scenario  for testing  tomography methods.

We did not apply any seasonal adjustment or other more complex dynamic models to these data sets, given the short time period they span.
We would recommend such an extension for time series spanning longer periods; indeed, even for data spanning only two days, usage patterns by time-of-day can be present.
However, we endeavor to compare our deconvolution algorithms on equal footing---our focus is dynamic deconvolution.
Thus, all methods are implemented with only local dynamics, without seasonal adjustment.

\subsection{Competing methods}
\label{sec:competing}

We tested locally IID and smoothed Gaussian methods \citep{Cao:Dav:Van:Yu:2000}, a Bayesian MCMC approach  \citep{tebaldiwest1998}, a simple gravity method \citep[e.g., see][]{fang:vard:zhan:2007}, a tomogravity method \citep{zhang:roug:lund:dono:2003}, the Gaussian state-space model developed for regularization in Section \ref{sec:regularization}, the  multilevel state-space model with na\"ive regularization developed in Section \ref{sec:sirm}, and the proposed two-stage estimation procedure with model-based regularization. %
All approaches were implemented in R with extensions in C to avoid computational bottlenecks (e.g., IPFP).
For the methods of \citet{Cao:Dav:Van:Yu:2000} and the Gaussian state-space model, which use windowed estimates, we selected a window width of 23 observations 
on the basis of prior work \citep{Airoldi:2003}. The final point-to-point traffic volume estimates were generally insensitive to a range of window sizes---this is largely attributable to the use of estimated ${\bm x}_t$'s instead of ${\bm \lambda}_t$'s for regularization.
For the Abilene data, we considered the alternative model-based regularization procedure for the dynamic multilevel model, based on a simple gravity model as detailed at the end of Section \ref{sec:calibration}.

For the approach of \citet{tebaldiwest1998}, we tested both the original implementation and our own modification in which (following the authors' original notation) $\lambda_j$ and $X_j$ are sampled with a joint Metropolis-Hastings step.
The proposal distribution for this step is constructed by first proposing uniformly along the range of feasible values for $X_j$ given all other values, then drawing $\lambda_j$ from its conditional posterior given the proposed $X_j$.
This greatly improves the efficiency of the MCMC sampler, leading to improved convergence (we observed multivariate Gelman-Rubin diagnostics reduced by approximately an order of magnitude) and better predictions.
These improvements allow us to compare the approach of \citet{tebaldiwest1998} on a more level playing field, focusing on the underlying model while mitigating computational issues.

For inference in the  proposed dynamic model, we used 1000 particles and 10 MCMC iterations (in the move step of the SIRM filter) per time point in all experiments.
We selected the former based on the number of effectively independent particles per time point, targeting a minimum of 10 in pilot runs.
The number of MCMC iterations was chosen as a balance between computational burden and particle diversity.
For the tomogravity method of \citet{zhang:roug:lund:dono:2003}, we set $\lambda$ to $0.01$.
Results were insensitive to the choice of $\lambda$ across a wide range of values, as previously reported \citep{zhang:roug:lund:dono:2003}.
These choices were kept consistent across experiments because they offered acceptable trade-offs and enabled a meaningful comparison of the competing methods.

\subsection{Performance comparison}
\label{sec:performance}

We summarize performance of the methods described above on all three data sets in Tables \ref{table:bell}, \ref{table:cmu}, and \ref{table:abilene}.
Each row corresponds to a method, and the columns provide mean $L_1$ and $L_2$ errors over time for the estimates of OD traffic in each data set with corresponding standard errors.
For the Bell Labs data set, we provide errors in kilobytes; for the CMU and Abilene data, we provide errors in megabytes.
We also provide Figure \ref{fig:bellLabsFitted} below and Figures S1--S5 in the supplemental material as a visualization of our results on the Bell Labs data set.
We compare and discuss performance in terms of accuracy, computational stability, and scalability.

\begin{table}[t!]
\begin{centering}
\begin{tabular}{l|cccc}
 & \multicolumn{4}{c}{\textbf{BELL LABS}} \\
\hline 
\textbf{Method}  & \textbf{$\bm{L}_{2}$ Error}  & \textbf{SE}  & \textbf{$\bm{L}_{1}$ Error}  & \textbf{SE}  \\
\hline 
Gravity & 62.96 & 3.16 & 182.58 & 7.69 \\
Tomogravity \citep{zhang:roug:lund:dono:2003} & 62.96 & 3.16 & 182.58 & 7.69 \\
Locally IID model  & 104.59 & 5.54 & 160.24 & 6.53 \\
Smoothed locally IID  & 104.25 & 5.52 & 157.87 & 6.48 \\
Tebaldi \& West (uniform prior)  & 76.60 & 4.91 & 173.94 & 7.49 \\
Tebaldi \& West (joint proposal){*}  & 49.43 & 2.58 & 147.66 & 6.18 \\
\hline 
Gaussian State-Space model            & 19.35  & 0.72  & 57.66  & 2.06  \\
Dynamic multilevel model (na\"{i}ve prior)  & 63.29  & 3.35  & 178.43  & 8.09  \\
Dynamic multilevel model (SSM prior)    & 19.93  & 0.87  & 58.20  & 2.39  \\
\end{tabular}
\par\end{centering}
\caption{Performance comparison with Bell Labs data, all results in KB. {*}
Denotes our own improvement on the original algorithm by Tebaldi \&
West.{ Note that the performance of simple gravity and tomogravity are identical on this network due to its star topology.
\label{table:bell}}}
\end{table}

\begin{table}[t!]
\noindent \begin{centering}
\begin{tabular}{l|cccc}
 & \multicolumn{4}{c}{\textbf{CMU}} \\
\hline 
\textbf{Method}  & \textbf{$\bm{L}_{2}$ Error}  & \textbf{SE}  & \textbf{$\bm{L}_{1}$ Error}  & \textbf{SE } \\
\hline 
Gravity  & 499.24 & 11.32 & 1521.66 & 30.09 \\
Tomogravity \citep{zhang:roug:lund:dono:2003} & 310.61 & 5.95 & 1096.38 & 18.68  \\
Locally IID model  & 592.49  & 9.91 & 1169.15 & 17.11  \\
Smoothed locally IID  & ---  & ---  & ---  & ---  \\
Tebaldi \& West (uniform prior)  & ---  & ---  & ---  & ---  \\
Tebaldi \& West (joint proposal){*}  & 167.94  & 4.42  & 712.37  & 14.68  \\
\hline 
Gaussian State-Space model  & 110.47  & 6.19  & 389.14  & 16.72  \\
Dynamic multilevel model (na\"{i}ve prior)  & 311.21  & 6.25  & 1109.68  & 19.58  \\
Dynamic multilevel model (SSM prior)  & 93.42  & 5.20  & 334.74  & 13.64  \\
\end{tabular}
\par\end{centering}
\caption{Performance comparison with CMU data, all results in MB. {*} Denotes
our own improvement on the original algorithm by Tebaldi \& West.{\label{table:cmu}}}
\end{table}

\begin{table}[t!]
\begin{centering}
\begin{tabular}{l|cccc}
 & \multicolumn{4}{c}{\textbf{ABILENE}} \\
\hline 
\textbf{Method} & \textbf{$L_{2}$ Error} & \textbf{SE} & \textbf{$L_{1}$ Error} & \textbf{SE} \\
\hline 
Gravity & 7.51 & 0.05 & 4.05 & 0.02 \\
Tomogravity \citep{zhang:roug:lund:dono:2003} & 5.26 & 0.05 & 3.06 & 0.02 \\
Locally IID Model & 12.17 & 0.07 & 7.03 & 0.03 \\
Tebaldi \& West (joint proposal){*} & 12.74 & 0.07 & 7.44 & 0.04 \\
\hline 
Gaussian State-Space model & 15.48 & 0.09 & 8.42 & 0.05 \\
Dynamic multilevel model (SSM prior) & 14.89 & 0.08 & 8.09 & 0.05 \\
Dynamic multilevel model (Gravity prior) & 4.01 & 0.03 & 2.49 & 0.01 \\
\end{tabular}
\par\end{centering}
\caption{Performance comparison with Abilene data, all results in MB. {*} Denotes
our own improvement on the original algorithm by Tebaldi \& West.{\label{table:abilene}}}
\end{table}

\begin{figure}[t!]
\begin{centering}
\includegraphics[width=\textwidth]{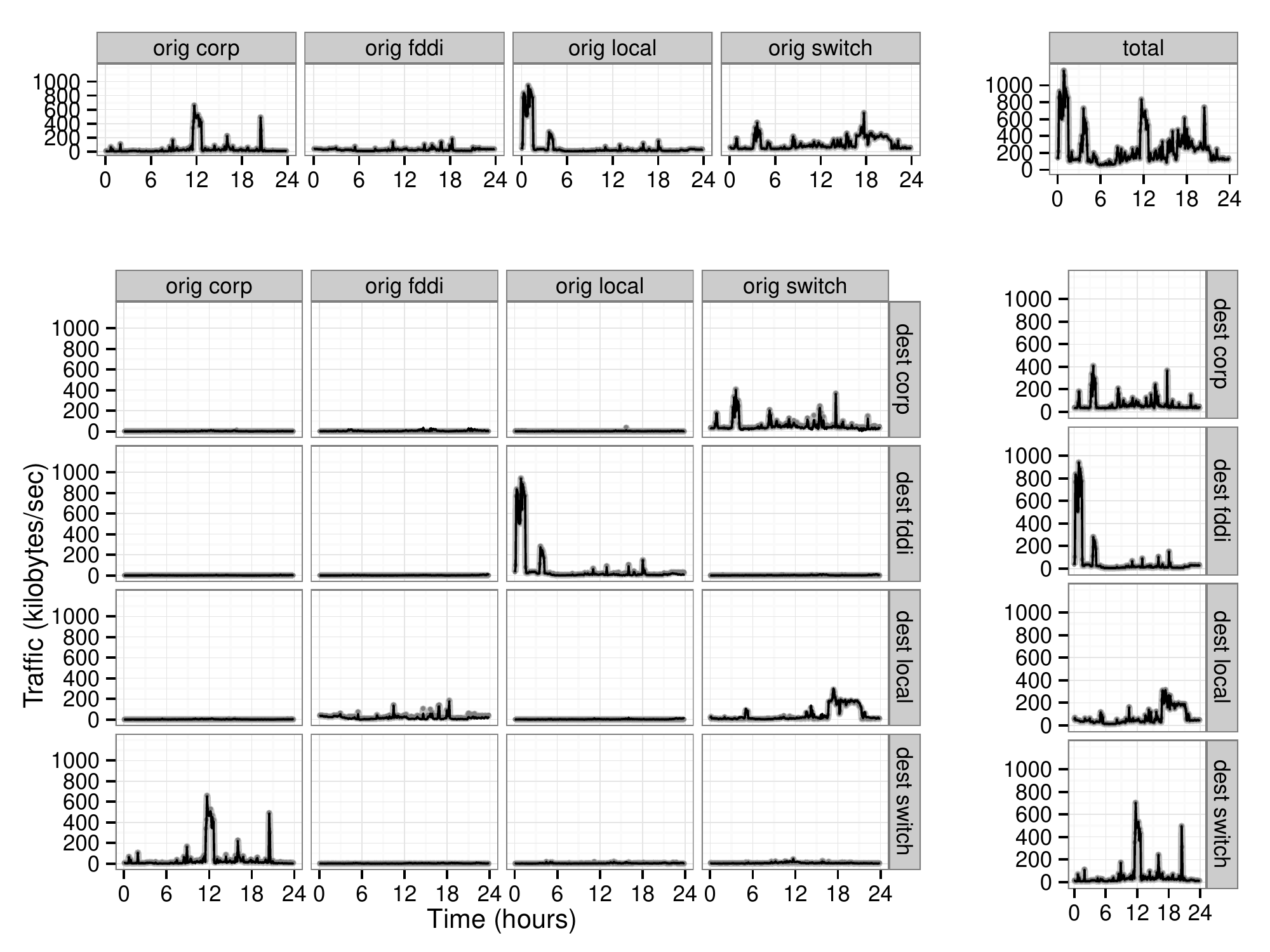}
\caption{Actual and fitted (dynamic multilevel model with SSM prior) traffic for Bell Labs data. Actual traffic is in gray, fitted traffic is in black. \label{fig:bellLabsFitted}}
\end{centering}
\end{figure}

\paragraph{Accuracy.}
We obtain favorable performance for the two-stage approach (corresponding to the bottom rows of Tables \ref{table:bell}--\ref{table:abilene}) for all three data sets.
For the Bell Labs data, mean (time-averaged) $L_1$ and $L_2$ errors are statistically indistinguishable (within 1 SE) between the calibration procedure and dynamic multilevel model with model-based regularization.
Both of these methods reduce average $L_1$ and $L_2$ errors by $60$-$80\%$ compared to the other approaches presented, representing a major gain in predictive accuracy.
For the CMU data, we obtain a reduction of $53\%$ in average $L_2$ error and $44\%$ in average $L_1$ error from the algorithm of \citet{Cao:Dav:Van:Yu:2000} to our multilevel state-space model; we observe $14$-$15\%$ reductions in average $L_1$ and $L_2$ errors from our Gaussian SSM to the multilevel state-space models.
Furthermore, we observe large gains in filtering performance for both data sets compared to inference using na\"ive regularization with our multilevel state-space model.
Overall, our approach outperforms existing methods in accuracy, with greater gains from the Gaussian SSM to the multilevel state-space model in our higher-dimensional setting.
The multilevel state-space model also outperforms gravity-based techniques substantially in these local area networks.

The comparative performance is somewhat different on the Abilene data.
Gravity and tomogravity perform very well on this backbone network, while the locally IID and Poisson models \citep{Cao:Dav:Van:Yu:2000,tebaldiwest1998} perform relatively poorly.
The same is true for the Gaussian SSM and the performance of the dynamic multilevel model suffers when using the SSM for regularization.
Using a simple gravity model for regularization, improves the performance of the dynamic multilevel model, leading to a reduction in mean $L_1$ and $L_2$ error of approximately $20\%$.
The variability in traffic volumes requires a smoothing based on instantaneous dynamics, rather than one based on a  constant scaling fixed across point-to-point routes.

A combination of three factors can be used to understand the performance of our methods, listed on the bottom three rows in Tables \ref{table:bell}--\ref{table:abilene}: explicit dynamics, heavy tails, and regularization. 
The multilevel state-space model fit with the na\"ive SIRM filter incorporates explicit dynamics and heavy tails, but its performance suffers because the distribution of $\bm \lambda$ is quite diffuse.
The network tomography inverse problem requires more constraints and outside information to yield useful solutions.
The Gaussian state-space model used for  calibration purposes is identifiable (as we show in Section \ref{sec:regularization}) and incorporates explicit dynamics, but does not account for heavy tails.
It performs well on the Bell Labs data set, where the distributions of origin-destination traffic are relatively symmetric (on the log scale), but suffers on the extremely heavy tailed CMU traffic volumes.
The multilevel state-space model fit with the two-stage procedure overcomes the ill-posedness of the underlying inverse problem and accounts for heavy tails, attaining comparable performance on the Bell Labs data and outperforming considerably on the CMU data as a result.
However, the Gaussian state-space model is not realistic for every setting.
Gravity methods offer a better source of regularization in backbone networks like Abilene.
This last finding agrees with  previous research \citep[e.g., see][]{zhang:roug:lund:dono:2003,zhang2009}.

\paragraph{Computational stability.}
We found a surprising amount of variation in computational stability among the methods evaluated.
The local likelihood methods \citep{Cao:Dav:Van:Yu:2000} remained stable across data sets, but the original method of \citet{tebaldiwest1998} encountered issues.
On the Bell Labs data, it required a very large number of iterations to obtain convergence (as indicated by the Gelman-Rubin diagnostic); 150,000 iterations per time were used to provide the given estimates, 50,000 of which were discarded as burn-in.
This method failed completely on several time points in the CMU data, becoming trapped in a corner of the feasible region.
Our revised version of the original MCMC algorithm performed better, requiring far fewer iterations for convergence; 50,000 iterations were sufficient for all examined cases, although 150,000 iterations were used for the results presented for comparability.

Our calibration procedure proved computationally stable across all three data sets.
The direct use of marginal likelihood, for maximum likelihood estimation, proved effective  in both the low- and high-dimensional data sets.
The multilevel state-space model was also stable in both settings; however, it proved to be sensitive to some of the alternative specifications mentioned in Section \ref{sec:estimation}.
In particular, major problems arose in experiments using the posterior on $\bm{x}_t$ from the previous time as a proposal (as is common in applications of particle filtering); several time points in the Bell Labs data required over 10 million proposals to obtain a single feasible particle. %
Additional care was needed with the ``move'' step due to similar issues.
Furthermore, the use of a na\"ive, random-walk regularization caused some computational difficulties, as the particles often became extremely diffuse in the feasible region.
Overall, we found inference with the multilevel state-space model computationally stable so long as sampling methods for highly constrained variables ($\bm{x}_t$ in particular) explicitly respected said constraints, proposing only valid values.
Our random directions algorithm (detailed in Algorithm \ref{alg:rda}) handles this task well.

\paragraph{Scalability.}
All methods we evaluated fared well in scalability, including the computationally-intensive, sequential SIRM inference we used for the multilevel state-space model.
On the Carnegie Mellon data set, for each time point, the methods of \citet{Cao:Dav:Van:Yu:2000} required approximately 225 seconds to obtain maximum likelihood estimates with a 23 observation window.
Our modification of the  sampler by \citet{tebaldiwest1998}required approximately 1500 seconds to obtain 150,000 samples for a single time points---the original MCMC sampler required 2250 seconds on average and often did not complete.
In contrast, the simulation-based filtering method for the multilevel state-space model required 270 seconds per time point on average, on the Carnegie Mellon data, and 210 seconds per time on average, on the Abilene data.
Approximately $70\%$ of this time was spent in the move step (MCMC) of the SIRM algorithm, with the vast majority of the remainder used for the random direction sampler.
On the Bell Labs data set, the filtering method required approximately 8 seconds per time, whereas our modification of the sampler by \citet{tebaldiwest1998} required 150 seconds per time---the original algorithm required approximately the same time.

These results are encouraging: the filtering algorithm is reasonably efficient (even written in R) and can run faster than real-time with 144 point-to-point traffic volumes at 5-minute sampling intervals.
We note that the Abilene data set required less computation per time than the Carnegie Mellon data set, even though both involved 144 point-to-point traffic time series.
This is because the effective dimensionality of the the ill-posed linear inverse problem is substantially lower for Abilene; we observe 24 linearly-independent aggregate traffic time series for Carnegie Mellon and 42 for Abilene, yielding 120 and 102 undetermined point-to-point traffic time series, respectively.
The reduction in undetermined dimensions closely tracks the reduction in computation for the  SIRM filter, as expected---the key computations of this sampler scale in complexity as the product of effective dimension and the number of point-to-point time series.
The proposed method takes advantage of the geometric structure of each data set to simplify sampling and guide inference.

Given more efficient implementation and parallelization, which are feasible for all sampling steps, the two-stage approach can scale to networks with many time more nodes.
This is especially true given the sparsity of the traffic on many such point-to-point routes; the prevalence of zero (observable) aggregate traffic volumes in real-world data further reduces the effective size of the deconvolution problem.
By more efficient implementation, we refer to the actual code for the SIRM filter.
In the \texttt{\small networkTomography} package, all parts of the SIRM filter itself are implemented in R.
Moving this algorithm to a compiled language (C or Fortran) and eliminating many memory allocations promises an order of magnitude speedup, based on initial development.

\section{Simulation studies}
\label{sec:simulations}

Here we further explore the relative performance of the methods we applied to the real data in Section \ref{sec:performance} by designing two experiments that involve a mix of simulated and real data.

We sought to understand the source of the performance of the competing inference methods with two experiments. 
In the first experiment, we simulated data from the model and compared the performance of the na\"ive random-walk prior with the two-stage estimation strategy.
The results show that the two-stage inference strategy leads to consistently better computational performance.
The second experiment involves a large-scale simulation study that compares the available methods by combining real origin-destination traffic with simulated network topologies.
We simulate a number of such scenarios by changing the degree of sparsity in the traffic and the complexity of the routing matrix, according to an experimental design that allows for an ANCOVA analysis.
The results show that significant error reduction can be expected by using the  two-stage estimation strategy.

\subsection{Evaluation of the two-stage inference strategy}
\label{sec:eval_2stage}

In the first experiment, we sought to quantify whether the two-stage estimation strategy proposed in Section \ref{sec:2stage} leads to consistently lower $L_1$ and $L_2$ errors, on average, over time and origin-destination routes.

We simulated origin-destination traffic from our multilevel state-space model under 3 network topologies: a 3-node bidirectional chain, a 3-node star topology, and a 4-node star topology,
 corresponding to 2, 4, and 9-dimensional solution polytopes for  inference on $\bm{x}_t$.
For each of these cases, we produced 30 data sets consisting of 300 time-points by drawing from the given multilevel model.
We drew all initial origin-destination traffic from log-Normal distributions with median 500 and geometric standard deviation 6.
Subsequent evolution of these traffic volumes was simulated with $\rho = 0.5$ and all other parameters as described in Section \ref{sec:constants}.
We then computed the implied aggregate traffic volumes for each replicate and fit the multilevel model to these data using the two-stage estimation strategy. 
In addition to the two-stage approach outlined previously, we also performed filtering using our multilevel state-space model with a na\"{i}ve random-walk regularization on the origin-destination traffic; that is, we set $\theta_{1 \, i,t} = 0 \quad \forall (i,t)$ and $\theta_{2 \, i,t} = \log(5)/2$.
This allows us to directly evaluate the effects of regularization and the plausibility of our model.

The primary quantity of interest in our simulations are the relative mean $L_2$ and $L_1$ errors in estimated origin-destination traffic for the na\"ive SIRM particle filter compared to our two-stage method.
The distributions of these relative $L_2$ errors is summarized in Figure \ref{fig:relerror}.
The magnitude of the errors is unchanged for the relative $L_1$ errors.
\begin{figure}[t!]
 \centering
  \includegraphics[width=0.75\textwidth]{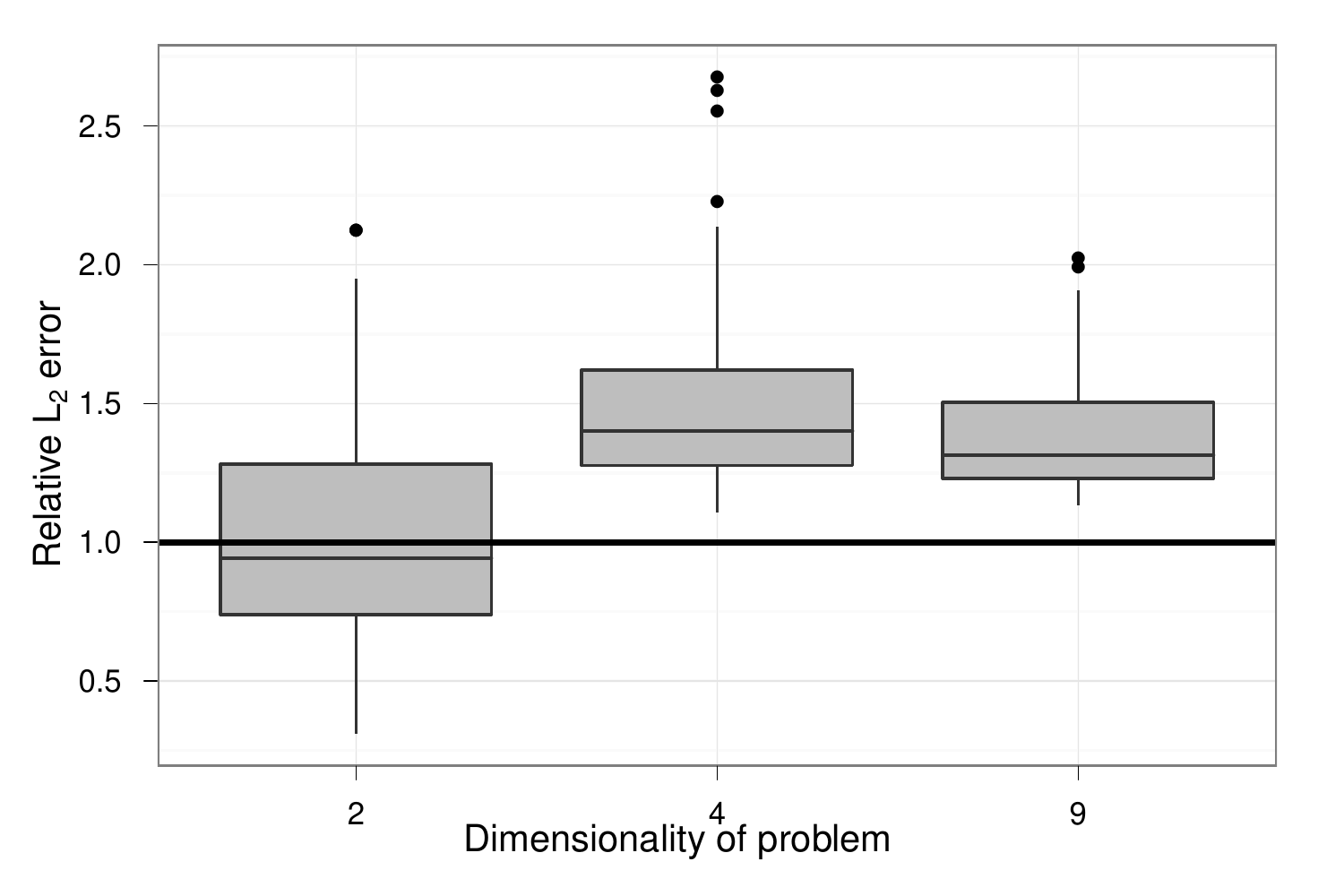}
  \caption{Relative $L_2$ error for na\"ive vs. two-stage method against dimensionality}
  \label{fig:relerror}
\end{figure}

We find that our two-stage method clearly outperforms the na\"ive SIRM particle filter when the dimensionality of the solution polytope is larger than two.
Specifically, we have a mean relative error of $1.09 \pm 0.49$ in two dimensions, increasing to $1.57 \pm 0.45$ in four dimensions and $1.40 \pm 0.26$ in nine dimensions.

Our experience with these simulations also highlighted the computational benefits of the proposed two-stage strategy.
During  iterations with the na\"ive SIRM filter, the {\em effective} number of particles rarely climbed above 2, whereas we typically obtained $10-50$ with the two-stage approach, with an equivalent {\em actual} number of particles.
With real data, we expect additional benefits from the two-stage estimation; in particular, we would expect it to have greater robustness to model misspecification.
Essentially, we are using information from a simpler model to rein-in potential issues with the more delicate multilevel model. This strategy is expected to stabilize inference in the latter and limit problems of non-identifiability.
These intuitions are further explored in Section \ref{sec:eval_methods}.

\subsection{Quantifying the factors that affect performance}
\label{sec:eval_methods}

In the second experiment, we sought to quantify the performance of the proposed method and existing methods relative to a the simple IPFP as baseline, explicitly controlling for the size for the network and the sparsity of the origin-destination traffic volumes over time.
 
To get a better sense of the relative performances on real data, we designed experiments using real traffic time series.
We selected a subset of the 1024 most active origin-destination traffic volumes from the CMU data set as the population of time series for designing these experiments.
The remaining CMU origin-destination pairs had negligible amounts of traffic.
We used realistic but artificially designed routing matrices based on star topologies with 3, 4, 5, and 9 nodes, and we generated 10 data sets for each topology by randomly sampling from the population of time series.

This combination of active origin-destination traffic volumes and star topologies defines a difficult scenario for inference underlying the network tomography problem.
In cases with extremely sparse OD traffic time series, we can deterministically infer many of them from aggregate traffic, since a large portion of measurements will be zero.
Hence, using high  OD traffic volumes reduces the number of cases with deterministic solutions and puts the methods to a stringent test.

The star topology also creates a stringent test for our method.
As noted by \citet{Singhal2010}, the star case is not worst with respect to, for instance, mean identifiability.
However, it does provide the fewest measured aggregate traffic time series for a given number of OD routes.
That is, given $d$ nodes with $n = d^2$ OD routes, and assuming all communications are bidirectional, any connected topology will have at least $2d$ aggregate time series.
Star topologies attain this lower bound, maximizing the dimension of the feasible region for OD traffic volumes given observed aggregate traffic.
This dimension is the relevant measure of difficulty for  inference in network tomography, so the star topology provides an appropriate benchmark.

On each generated data set, we ran five methods to estimate the OD traffic volumes: IPFP, the locally IID method of \citet{Cao:Dav:Van:Yu:2000}, our implementation of the Poisson model from \citet{tebaldiwest1998}, and the multilevel  model with a na\"ive random-walk prior, the proposed calibration procedure, and the proposed two-stage inference strategy.

The outcomes of interest are average $L_1$ and $L_2$ errors, over time and origin-destination routes.
We performed an ANCOVA analysis using the average errors from each of our experiments to understand what drives performance in this problem.
The primary factor of interest for this analysis is the method used.
We coded network size as a factor with three levels.
We included sparsity as a covariate as well to capture the effect of having deterministically zero observed traffic as described above.
Sparsity enters the ANCOVA analysis as $\log_{10}($average proportion of measured traffic volumes that are not deterministically 0$)$.

\begin{table}[t!]
\caption{Log-linear ANCOVA model for simulation study; $\log_{10}(L_1$ errors$)$ is outcome. \label{table:simAncova}}
\label{tab:ancova}
\begin{center}
\begin{tabular}{rrrrr}
  \hline
                       & Estimate & Std. Error & t value & Pr($>$$|$t$|$) \\ 
  \hline
(Intercept)            &  6.526 & 0.115 & 56.60 & 0.000 \\ 
  Dim=4                &  0.058 & 0.111 &  0.53 & 0.599 \\ 
  Dim=5                &  0.908 & 0.108 &  8.42 & 0.000 \\ 
  Dim=9                &  2.103 & 0.108 & 19.47 & 0.000 \\ 
  Locally IID method   &  2.205 & 0.130 & 17.03 & 0.000 \\ 
  Tebaldi \& West      & -0.116 & 0.130 & -0.89 & 0.373 \\
  Na\"ive prior        & -0.021 & 0.130 & -0.16 & 0.870 \\ 
  Calibration model    & -0.241 & 0.130 & -1.86 & 0.065 \\ 
  Two-Stage Inference  & -0.244 & 0.130 & -1.88 & 0.061 \\  
  $\log_{10}$ sparsity & -8.237 & 3.179 & -2.59 & 0.010 \\ 
   \hline
\end{tabular}
\end{center}
\end{table}

Table \ref{tab:ancova} summarizes the results of this analysis for $L_1$ errors.
We used a log-linear model for this analysis; initial diagnostics suggested that its variance structure is more appropriate for this experimental data than an untransformed model.
We checked for interactions between dimensionality and method, but found no support in the data ($p=0.996$ for a standard F test).
It appears that any such interactions would require larger networks to identify.

We find that the proposed methods significantly outperform the baseline IPFP approach, while the locally IID method performs significantly worse. The performance of the Bayesian model by \citet{tebaldiwest1998} and the multilevel methods fit with the na\"ive SIRM filter are inconclusively better then IPFP.
The multilevel model consistently under-performs the model of \citet{tebaldiwest1998}, as well as our other approaches, when used with na\"ive random-walk priors.
However, the calibration procedure alone performs quite well despite its simplicity.
The performance of the proposed  two-stage approach is not significantly higher than that for the calibration procedure in this setting, which suggests that our calibration procedure is the driver of our performance improvements at this scale.
This agrees with our empirical findings with the Bell Labs data set and is compatible with the observed increase in performance at larger scales, on the CMU data set, in Table \ref{table:cmu}.
These results are essentially unchanged (qualitatively and quantitatively) when we substitute $L_2$ for $L_1$ errors.

\section{Concluding remarks}
\label{sec:remarks}

In this paper, we address the problem of (volume) network tomography in a dynamic filtering setting.
For this application, we develop a novel  approach to this problem by combining a new multilevel state-space model that posits non-Gaussian marginals and nonlinear probabilistic dynamics with a novel two-stage inference strategy.
Our results and analyses substantiate several claims and suggest points for further discussion.

We analyzed three networks (Bell Labs, Carnegie Mellon, and Abilene) which span a wide range of dimensions, with different inference methods.
The results demonstrate a clear improvement of the proposed methodology over previously published methods in estimating point-to-point traffic volumes.
Comparison between Bell Labs and Carnegie Mellon results suggests that this gain increases with the dimensionality of the problem.
Our results with the Abilene network highlight the differences between local-area  and Internet2 backbone networks.
They differ in both topology and traffic dynamics, requiring different approaches to regularization.

\subsection{Modeling choices}
\label{sec:remarks_model}

Our model explicitly captures two critical feature of point-to-point traffic---namely, skewness and temporal dependence.
The substantial improvements in accuracy over existing methods can be attributed to these modeling improvements, to a large extent.
The gains in computational efficiency are responsible for the improvements in accuracy only in part, as we discuss below.
Previous modeling approaches have accounted for skewness \citep{tebaldiwest1998}, but never for explicit temporal dependence of the point-to-point traffic volumes.
The inter-temporal smoothing algorithm of \citet{Cao:Dav:Van:Yu:2000} includes elements of temporal dependence; however, the model assumes temporally independent time series and the dependence is imposed indirectly having observations  within a time window  that contribute to inference at any given time point.
In summary, previous work has not accounted for the range of properties addressed by our model.
The performance gains that stem from our modeling assumptions are clear on the three data sets tested; in particular, the gains from the model of \citet{tebaldiwest1998} to the Gaussian SSM and to the dynamic multilevel model for the CMU data set reinforce the benefits of positing realistic dynamics  in this problem.

We chose to increase the probability of near-zero traffic volumes using a truncated Gaussian error, rather than a log-Normal or Gamma distribution whose support is naturally on the non-negative reals.
From a computational perspective, given that the particle filter involves a Metropolis step, the truncated Gaussian error is not particularly more tractable than log-Normal or Gamma errors.
However, the truncated normal increases the probability assigned to any $[0,\epsilon)$  interval, relatively to Gamma or log-Normal with same mean and scale.
To obtain similar behavior with these non-truncated noise structures would either be impossible, e.g., with a log-Normal distribution, or require ad-hoc re-parametrization linking the shape (intuitively, variability) and center (magnitude of point-to-point traffic volumes), e.g.
with a Gamma distribution.
We believe that using a truncated distribution with the mode calibrated using the underlying intensity process ($\bm \lambda_t$ in our notation) provides a cleaner solution.
However,  several alternatives are viable.

Fundamentally, we estimate the point-to-point traffic volumes by projecting  aggregate traffic volumes onto the latent space point-to-point traffic inhabits; that is, we want to compute $\mathbb{E}[\bm x_t \mid \bm y_t]$ under a given probabilistic structure.
The relative variability of origin-destination flows over time plays a large role in inference, as there is typically a strong relationship between the mean and variance of point-to-point traffic.
Because of this, simple methods that do not model variability explicitly and realistically, including Moore-Penrose generalized inverse \citep{harv:2008} and independent component analysis \citep{Hyva:Karh:Oja:2003}, are of limited use in this context.
Surprisingly, however, we found that the iterative proportional fitting procedure \citep{fien:1970} often produces  reasonable solutions, likely because each such solution corresponds to a feasible set of origin-destination flows. \citet{liantaftyu:2006} have recently capitalized on this finding.
In contrast, our approach, models this variability with a probabilistic structure, improving inference by using this additional information.

\subsection{Applicability to more complicated routing schemes}
\label{sec:remarks_extensions}

Routing schemes other than deterministic routing, such as probabilistic routing and dynamic load-balancing, 
can be subsumed within the modeling framework we developed in Section \ref{sec:model}. 

A probabilistic routing scheme would be captured by a probabilistic $A$ matrix. 
 Column $j$ of the $A$ matrix specifies the $m$ proportions of point-to-point traffic volume $j$ that is measured at the $m$ counters.
In deterministic routing, each point-to-point traffic volume either contributes to a counter or does not, $A_{ij} \in \{0,1\}$.
In probabilistic routing,  each point-to-point traffic volume contributes to multiple counters with different probabilities, $A_{ij}\in [0,1]$.

Routing schemes that carry out dynamic load balancing to manage congestion would be captured by a time- and traffic-dependent matrix, $A(t)$.
However, congestion can only be monitored at the router level, in practice, using the {\em observed} aggregate traffic counters.
Thus the counters can be considered as covariates, and the routing matrix can be modeled as a function of these covariates,  \[
 A(t) = 
 \left\{
  \begin{array}{rl} 
   A^{High} & \hbox{ if } \bm y_t \hbox { is high} \\
   A^{Low} & \hbox{ if } \bm y_t \hbox { is low}, 
  \end{array}
 \right.
\]
More nuanced specifications are possible.
The key point is that implemented routing and switching protocols can only make use of the measurements collected at the router level, $\bm y_{t}$.

\subsection{Computational challenges and scalability}
\label{sec:remarks_computing}

Our inference method is computationally efficient and scales to larger inference problems than have been previously addressed.
The problem is fundamentally $O(n-m)$ for each time point, so we cannot hope to do better than quadratic scaling in the number of nodes $d$ in our network (excepting cases where a few aggregate traffic measurements are zero).
Despite the sophistication of our dynamic multilevel model, the sequential Monte Carlo technique allows for inference in better than real-time for a network with $144=n\approx O(d^2)$ origin-destination routes and $26=m\approx O(d)$ router level measurements.
As SIRM is the portion of the inference algorithm that would be used in an online application, we have demonstrated a scalable technique for inference with a model of greater complexity and realism than has been previously found in the literature.

These gains in computational efficiency also reduce numerical instability and are ultimately responsible for additional gains in accuracy.
Computational issues can be appreciated by considering the amount of effort needed to maintain $\Cov \bm e_t$ positive-definite in the EM algorithm of \citet{Cao:Dav:Van:Yu:2000}, especially when the traffic approaches zero.
We can see some artifacts in the corresponding point-to-point traffic  estimates in Supplemental Figure S1 (green lines) due to this issue in low traffic OD routes, e.g., see ``orig local $\rightarrow$ dest local''.
We further quantified the effects of computational efficiency on inference in the original methods by \citet{tebaldiwest1998} in Table \ref{table:bell} by comparing the uniform prior and component-wise proposal to the joint proposal we developed.
In addition to the gains in speed and convergence discussed in Section \ref{sec:performance}, we observe a large reduction in average error from the component-wise to joint proposal ($35\%$ in $L_2$ error, $15\%$ in $L_1$) which correspond to no changes in priors nor to the underlying model.

The random directions algorithm %
plays an important role 
in the sequential Monte Carlo sampler.
Without such an algorithm to sample directly from the feasible region of point-to-point traffic volumes, we would be forced to use a na\"ive proposal distribution.
In our testing, such distributions proved extremely problematic (as discussed in Section \ref{sec:estimation}), especially in higher dimensional settings.
In such cases, intelligent sampling techniques that fully utilize the geometric constraints implied by the data are necessary to obtain high accuracy and efficiency.
This is particularly salient comparing the results presented here to our previous work \citep{Airo:Falo:2004}; the method presented there suffered from computational instabilities.
It was hampered both by inefficient sampling on the feasible space of solutions and by  distributional assumptions that assigned low probability to low point-to-point traffic volumes.

Multi-modality of the marginal posterior on point-to-point traffic volumes $x_{it}$ appears negligible in our analyses.
Our theoretical results suggest that multi-modality in these problems is limited to that due to flat regions in the case of real-valued traffic volumes and  models assuming independent traffic.
We suspect that the issues with multi-modality discussed in the literature are due to mainly to the inefficiency of the samplers.
This further reinforces the importance of efficient computation for inference in highly complex, poorly identified settings; even a simple model can falter on poor computation, and complex models require great care to obtain reliable inference results.

\subsection{A novel two-stage inference strategy in dynamic multilevel models}
\label{sec:remarks_2stage}

As previously argued by \citet{tebaldiwest1998} in the static setting, informative priors are essential to identify the peak in the likelihood that correspond to the true configuration of point-to-point traffic.
This conclusion holds in the dynamic setting we consider, despite the additional information that temporal dependence makes relevant for the inference of  point-to-point traffic volumes.
The technical choices at issue are:
 (i) where to find such information---it is not obvious in the data;
 (ii) what parameters are most convenient to put priors on; and
 (iii) how do we translate the additional information into prior information for the chosen parameterization.

We use a simple identifiable model %
to find rough estimates of the point-to-point traffic volumes (in our first stage).
These estimates provide some information about where the point-to-point traffic volumes live in the space of feasible solutions, enabling us to identify high-probability subsets of the feasible region at each time point before embarking on computationally-intensive sampling.
The expected benefits from this strategy are larger in higher dimensions, as the proportion of the feasible region's volume with high posterior density decreases rapidly with dimensionality (the classical curse of dimensionality).
Practically, informative priors increase the efficiency of the particle filter by focusing the sampler on promising regions of the parameter space, avoiding wasted computation and improving inference.

In order to pass the first-stage information to the (non-Gaussian) dynamic multilevel model, we moved away from a standard linear state-space formulation with additive error to a non-linear formulation with stochastic dynamics, which effectively provides a multiplicative error (second-stage).
The stochastic dynamics assumed for $\bm \lambda_t$ provide our parameters of choice for incorporating information obtained in the first stage of estimation. Prior calibration for the dynamics of $\bm \lambda_t$ guides inference without placing too tight of a constraint on the inferred point-to-point traffic volumes.
In Section \ref{sec:calibration} we describe how we solve the problem of  translating the first-stage estimates into priors for the parameters of the second-stage model.
Essentially, we trade-off the need to pass as much information as possible from the first stage of estimation to the second with the controlled inaccuracy of the first-stage modeling assumptions.

Our two-stage procedure suggests a more general strategy for inference in dynamic hierarchical models with weak identifiability constraints.
The use of simpler models to guide inference in more sophisticated, realistic models, and (if necessary) to provide regularization, can result in large  performance gains, as we have demonstrated here.
This strategy implements a principled approach to {\em cutting corners} \citep{meng:2010}.
Ultimately, the proposed two-stage approach has allowed us to use a sophisticated generative model for inference, leveraging the power of multilevel analysis, while maintaining efficiency for real-time applications.

\newpage
\phantomsection
\section*{Acknowledgments}
\addcontentsline{toc}{section}{Acknowledgments}

The authors wish to thank Jin Cao and Matthew Roughan for providing point-to-point traffic volumes, and the referees for providing valuable suggestions.
This work was partially supported 
by the Army Research Office Multidisciplinary University Research Initiative under grant 58153-MA-MUR, and  
by the National Science Foundation under grants IIS-1017967, DMS-1106980 and CAREER IIS-1149662, all to Harvard University.
The views and conclusions contained in this document are those of the  authors and should not be interpreted as representing the official policies, either expressed or implied, of the Army Research Office, the National Science Foundation, or the U.S. government.

\clearpage

\bibliographystyle{plainnat}

\appendix
\clearpage
\section{Efficient inference for the Gaussian state-space model}
\label{supp:ssm}

Inference on the latent point-to-point time series $\bm x_t$ in the Gaussian state-space model specified by Equation \ref{eq:gssm} can be carried out with standard Kalman filtering and smoothing.
Estimating the constants underlying the model via maximum likelihood can be approached with two strategies: Expectation-Maximization \citep[EM,][]{Demp:Lair:Rubi:1977} and direct numerical optimization. 
The EM approach for unconstrained Gaussian state-space models  requires Kalman smoothing for the E-step and maximization of the expected log-likelihood for the M-step \citep{gh1996em}. While the E-step is straightforward and efficient to calculate using standard algorithms, the M-step requires expensive numerical optimization in our case. Due to the constraints on $Q$ and $\Cov \bm e_t$ and the dependence of the observations, there is no analytic form for the maximum of the expected log-likelihood.
Since the EM requires numerical optimization, we decided to use direct numerical optimization on the marginal likelihood obtained from the Kalman smoother.
This amounts to maximizing
\[
\textstyle
 \ell(Y \mid \theta) = -\sum_t \log | \hat \Sigma_t | -
   \frac{1}{2} \sum_t (\bm y_t - \bm{\hat y_t})'\hat \Sigma_t^{-1}(\bm y_t - \bm{\hat y_t}) 
\]
where $\bm{\hat y_t}$ and $\hat \Sigma_t$ are the estimated mean and covariance matrices from the Kalman smoother. With a fast (Fortran) implementation of the Kalman iterations, this approach yields favorable run-times and stable results.

Such efficient computation is, however, sensitive to certain modeling decisions.
Enforcing a steady-state starting value within each window is particularly useful.
Formally, suppose that we index each window of width $w$ with $t = 1, \ldots, w$.
We must specify a starting value $\bm x_0$ for each window.
By linking $\bm x_0$ to $\bm \lambda$, we can simplify our computation.
For a given choice of $\bm \lambda$, the steady-state mean of the process specified in Equation \ref{eq:gssm} is $\frac{1}{1-\rho} \bm \lambda$.
Fixing $\bm x_0 = \frac{1}{1-\rho} \bm \lambda$ allows us to reduce the dimensionality of Equation \ref{eq:gssm}.
Formally, we can rewrite it as
\begin{eqnarray}
    \bm{x}_t - \frac{1}{1-\rho} \bm \lambda & = &  F (\bm{x}_{t-1} - \frac{1}{1-\rho} \bm \lambda) + \bm{e}_t \\
     \bm{y}_t & = & A  (\bm{x}_t + \frac{1}{1-\rho} \bm \lambda) + \bm{\epsilon}_t \ . \nonumber  \vspace{3pt}
\end{eqnarray}
This reduces the dimensionality of our state variable by a factor of two, greatly accelerating all Kalman filter and smoother calculations.
As said calculations have complexity quadratic in the problem's dimensionality, this reduces the computational load by approximately a factor of 4.

\section{Actual vs. fitted traffic for Bell Labs and CMU data}

\renewcommand{\thefigure}{S\arabic{figure}}
\setcounter{figure}{0}

In this appendix, we show the actual vs. fitted OD flows for the methods presented previously. We plot all OD flows for the Bell Labs data and the 12 most variable OD flows for CMU. Ground truth is always in black, with estimated values in color. Figures \ref{suppfig:cao1router} through \ref{suppfig:naiive1router} cover the Bell Labs data, and Figures \ref{suppfig:caoCMU} through \ref{suppfig:naiiveCMU} cover the CMU data.

\begin{figure*}[h]
\centering
\includegraphics[width=7in,angle=90]{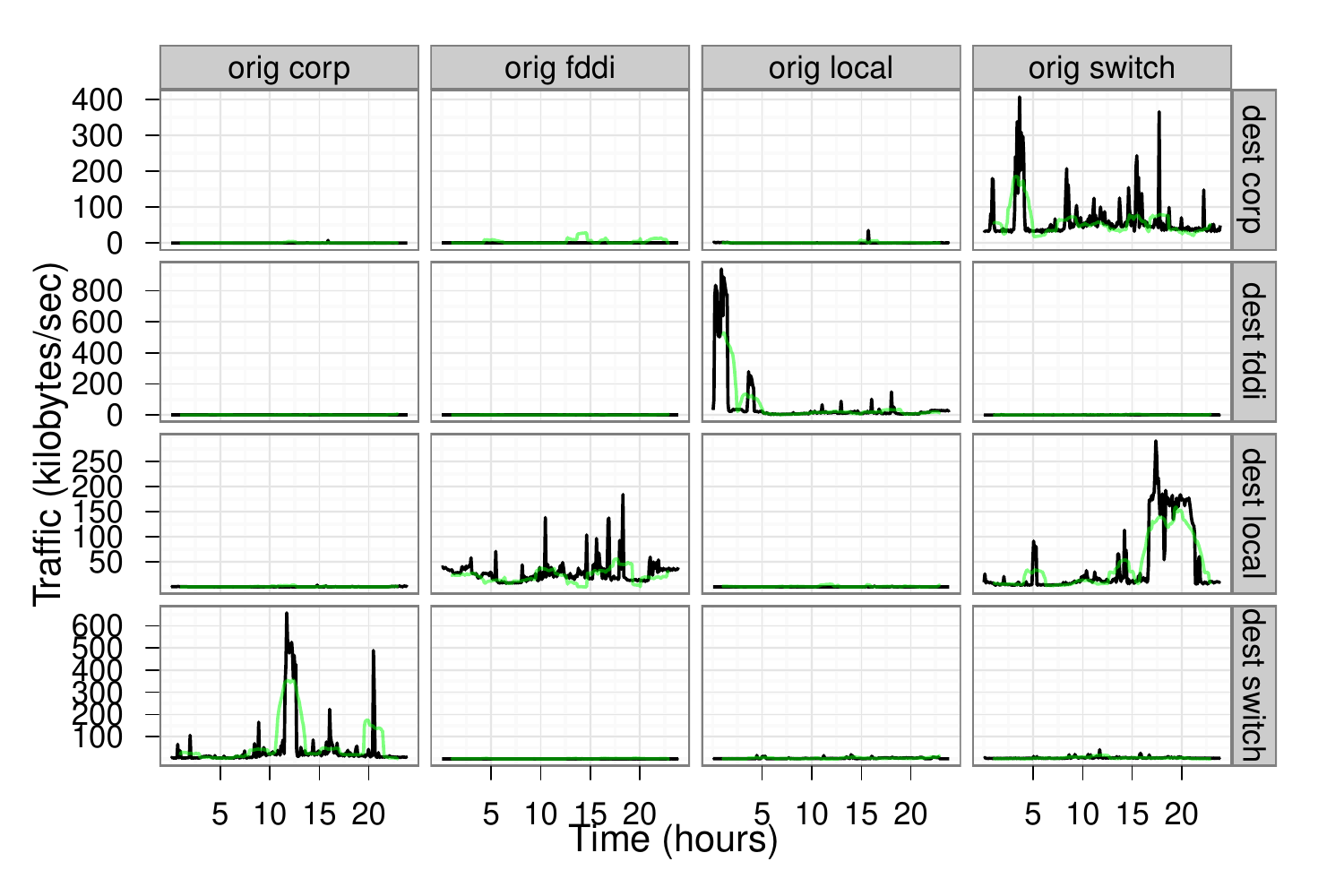}
\caption{Fitted values vs. ground truth for Bell Labs data. Ground truth in black; Locally IID model in green.}
\label{suppfig:cao1router}
\end{figure*}

\begin{figure*}[h]
\centering
\includegraphics[width=7in,angle=90]{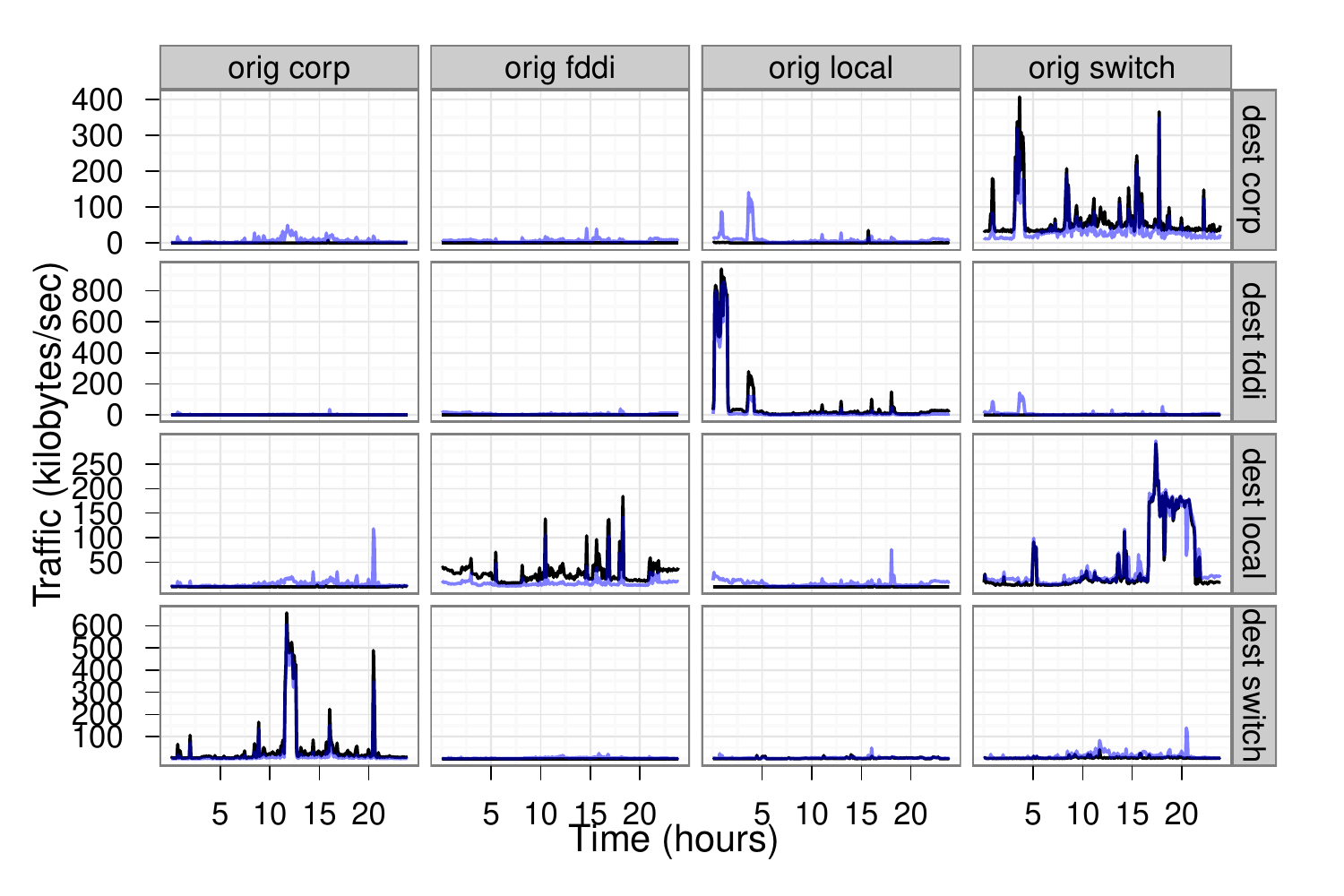}
\caption{Fitted values vs. ground truth for Bell Labs data. Ground truth in black; Tebaldi \& West (joint proposal) in blue.}
\label{suppfig:tw1router}
\end{figure*}

\begin{figure*}[h]
\centering
\includegraphics[width=7in,angle=90]{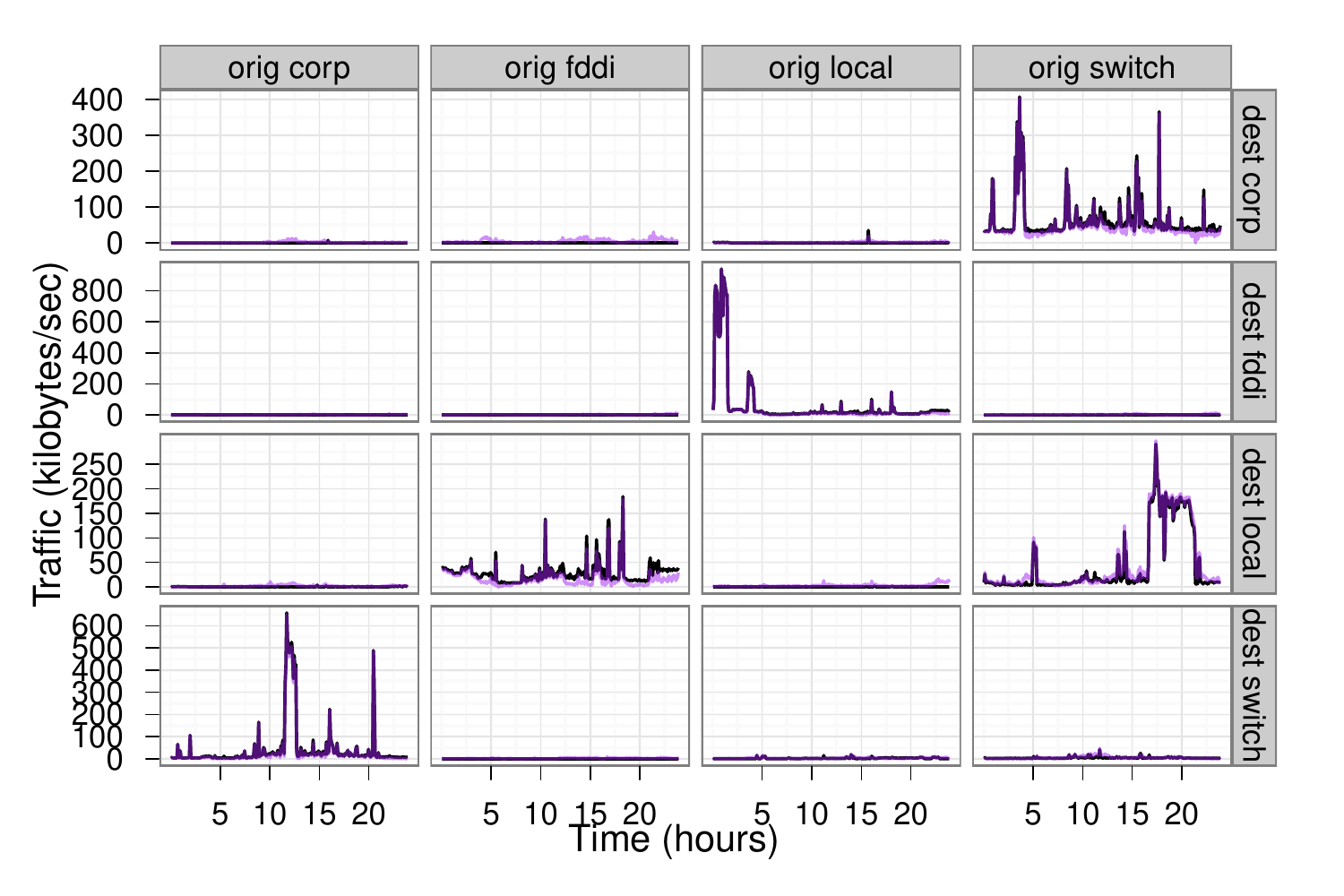}
\caption{Fitted values vs. ground truth for Bell Labs data. Ground truth in black; Calibration model (stage 1) in purple.}
\label{suppfig:ssm1router}
\end{figure*}

\begin{figure*}[h]
\centering
\includegraphics[width=7in,angle=90]{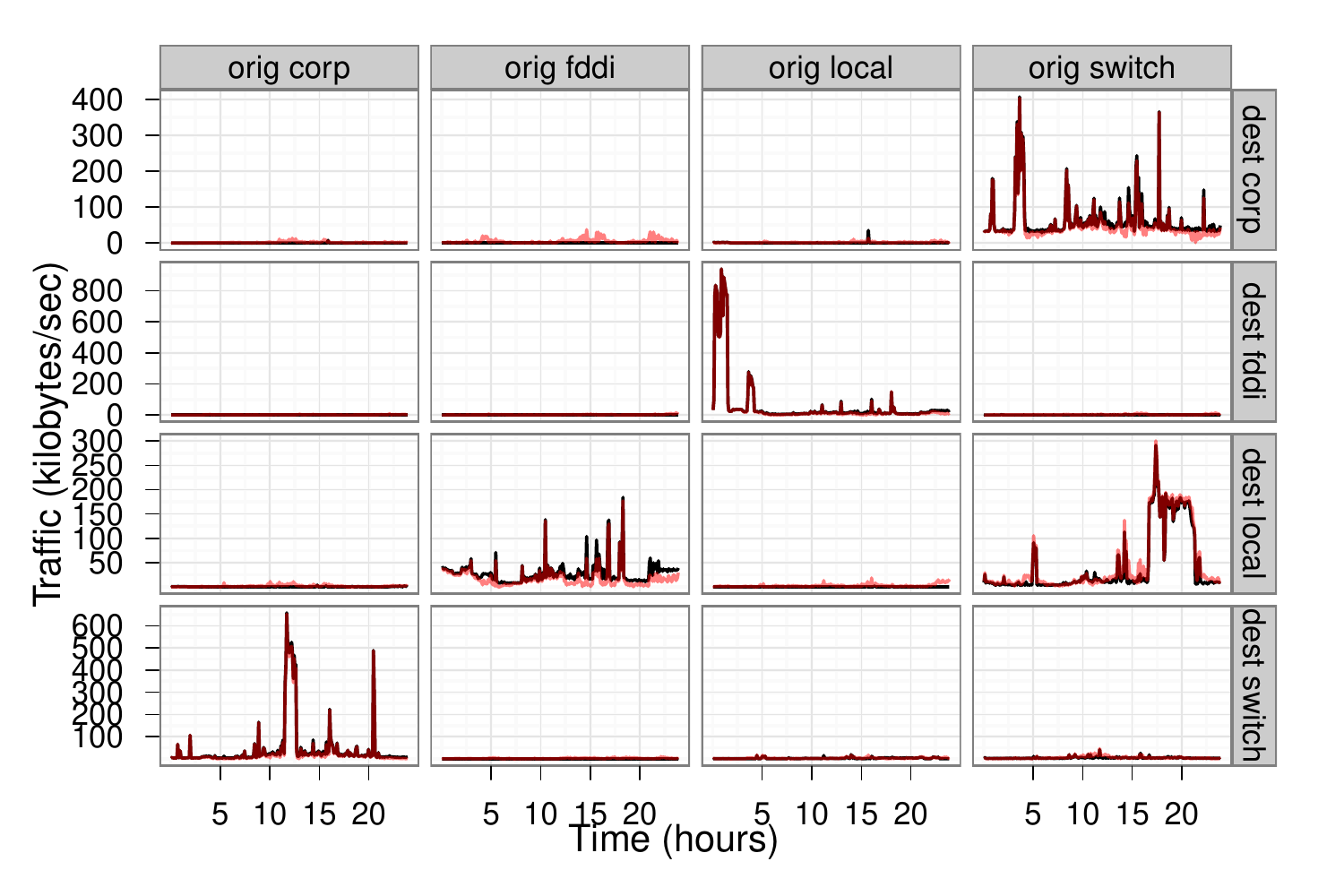}
\caption{Fitted values vs. ground truth for Bell Labs data. Ground truth in black; Dynamic multilevel model (stage 2) in red.}
\label{suppfig:dynamic1router}
\end{figure*}

\begin{figure*}[h]
\centering
\includegraphics[width=7in,angle=90]{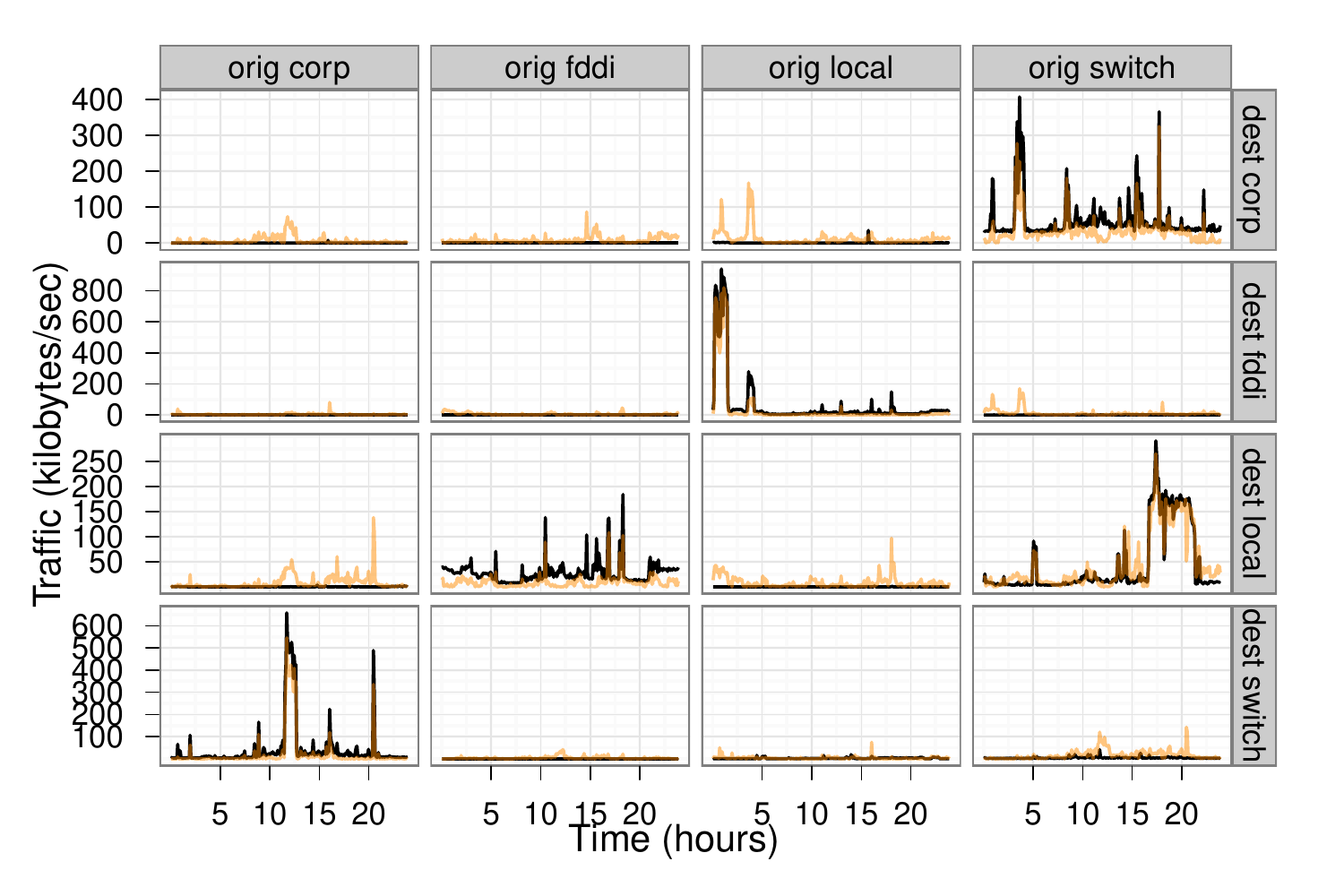}
\caption{Fitted values vs. ground truth for Bell Labs data. Ground truth in black; Na\"ive prior in orange.}
\label{suppfig:naiive1router}
\end{figure*}

\begin{figure*}[h]
\centering
\includegraphics[width=7in,angle=90]{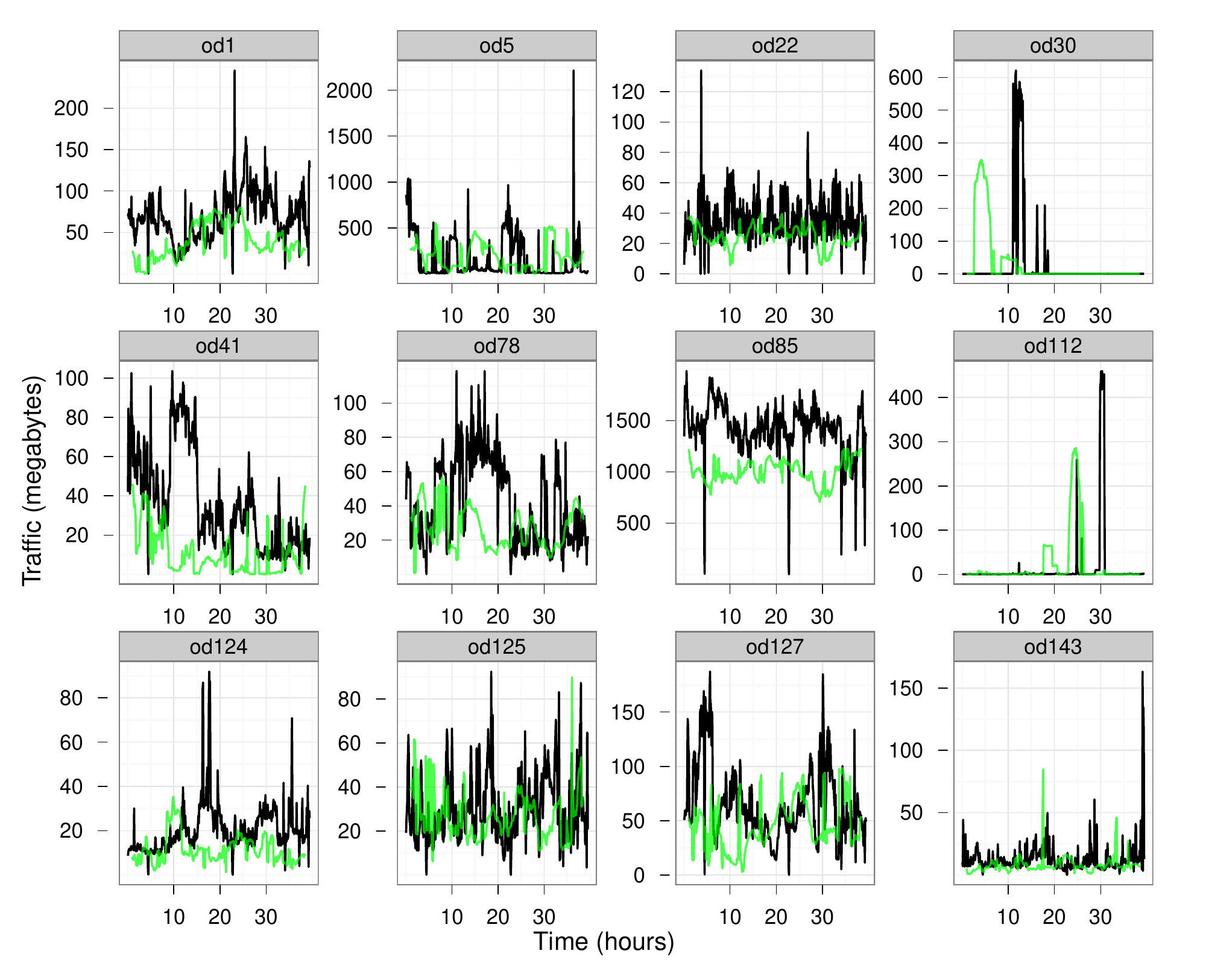}
\caption{Fitted values vs. ground truth for CMU data. Ground truth in black; Locally IID model in green.}
\label{suppfig:caoCMU}
\end{figure*}

\begin{figure*}[h]
\centering
\includegraphics[width=7in,angle=90]{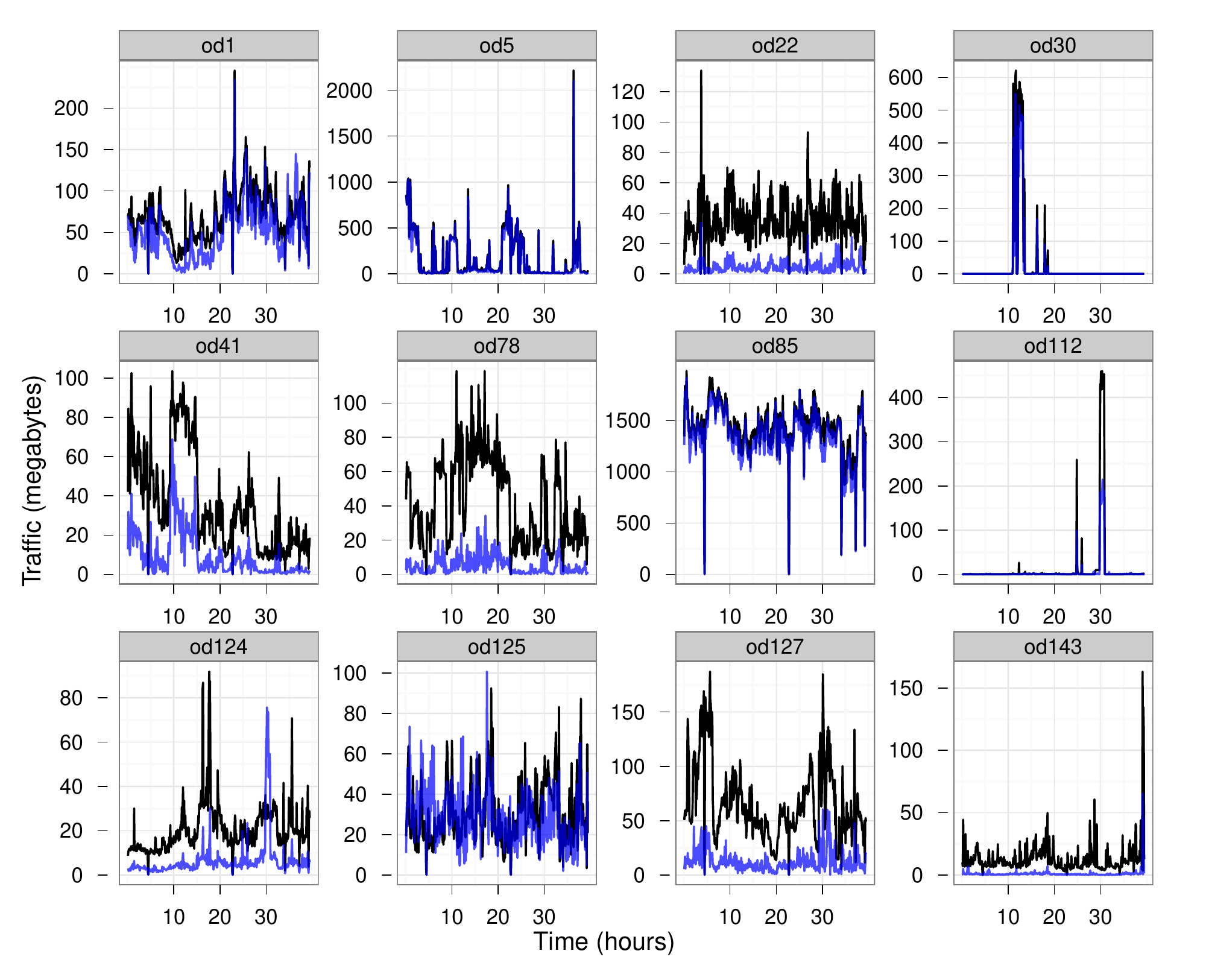}
\caption{Fitted values vs. ground truth for CMU data. Ground truth in black; Tebaldi \& West (joint proposal) in blue.}
\label{suppfig:twCMU}
\end{figure*}

\begin{figure*}[h]
\centering
\includegraphics[width=7in,angle=90]{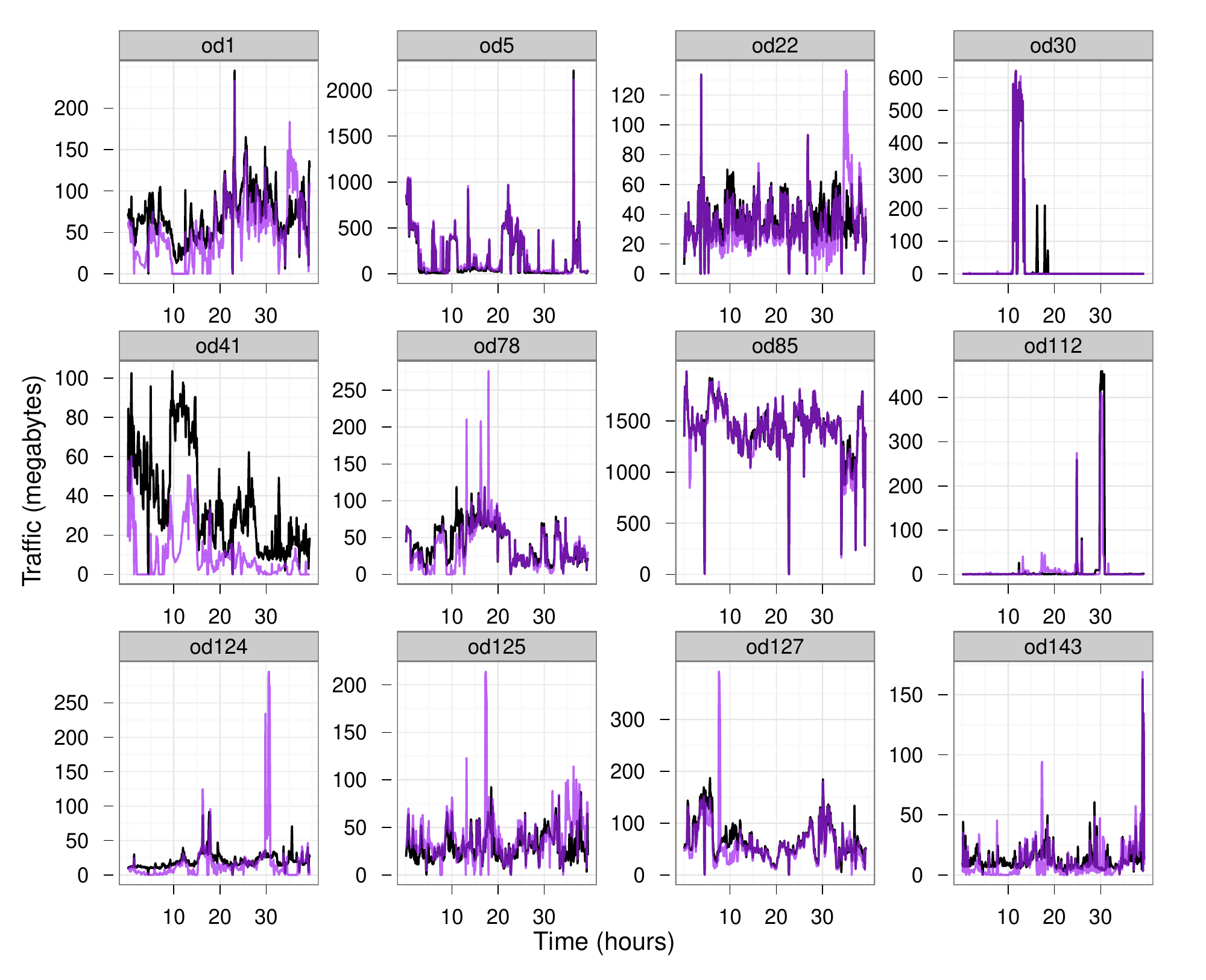}
\caption{Fitted values vs. ground truth for CMU data. Ground truth in black; Calibration model (stage 1) in purple.}
\label{suppfig:ssmCMU}
\end{figure*}

\begin{figure*}[h]
\centering
\includegraphics[width=7in,angle=90]{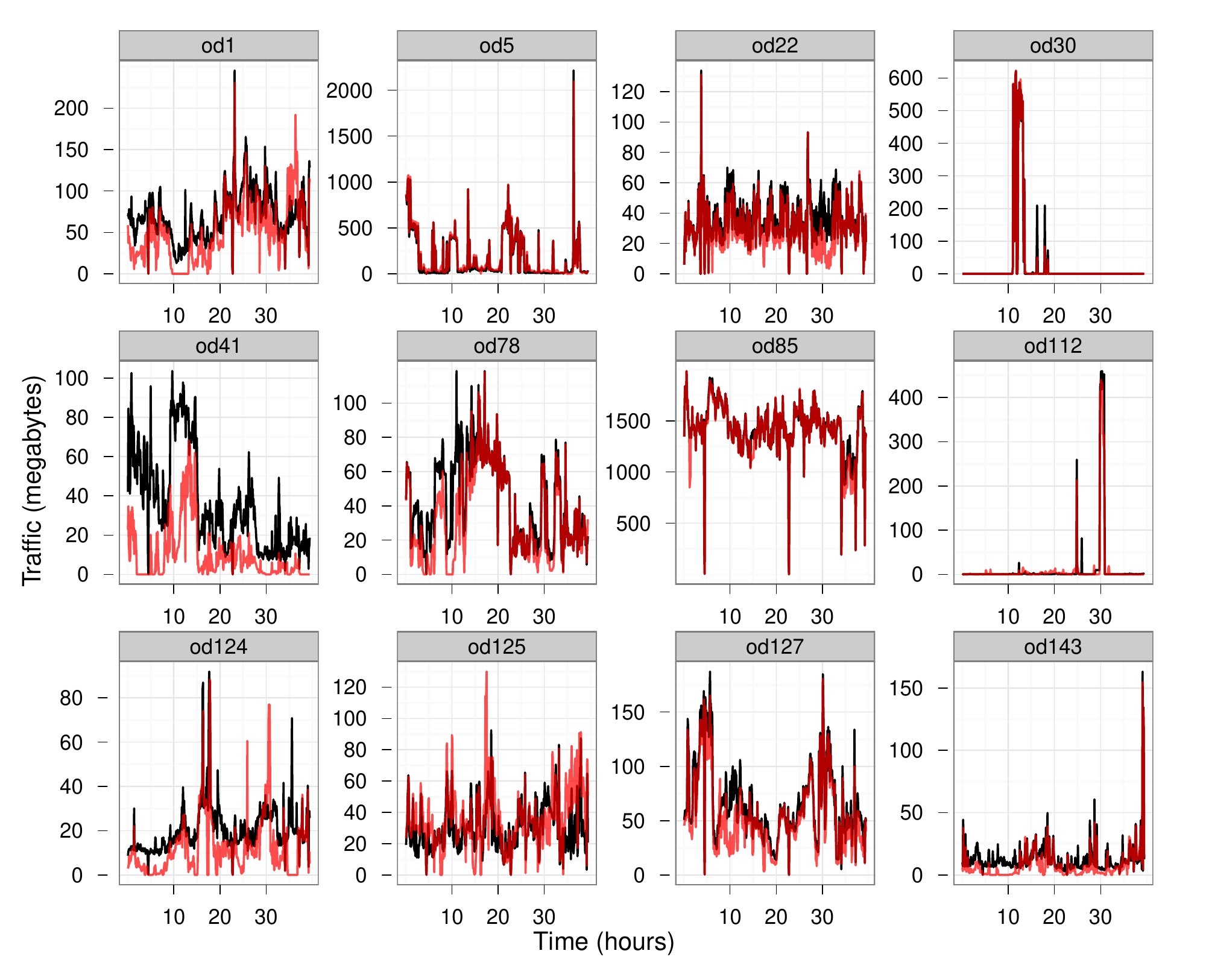}
\caption{Fitted values vs. ground truth for CMU data. Ground truth in black; Dynamic multilevel model (stage 2) in red.}
\label{suppfig:dynamicCMU}
\end{figure*}

\begin{figure*}[h]
\centering
\includegraphics[width=7in,angle=90]{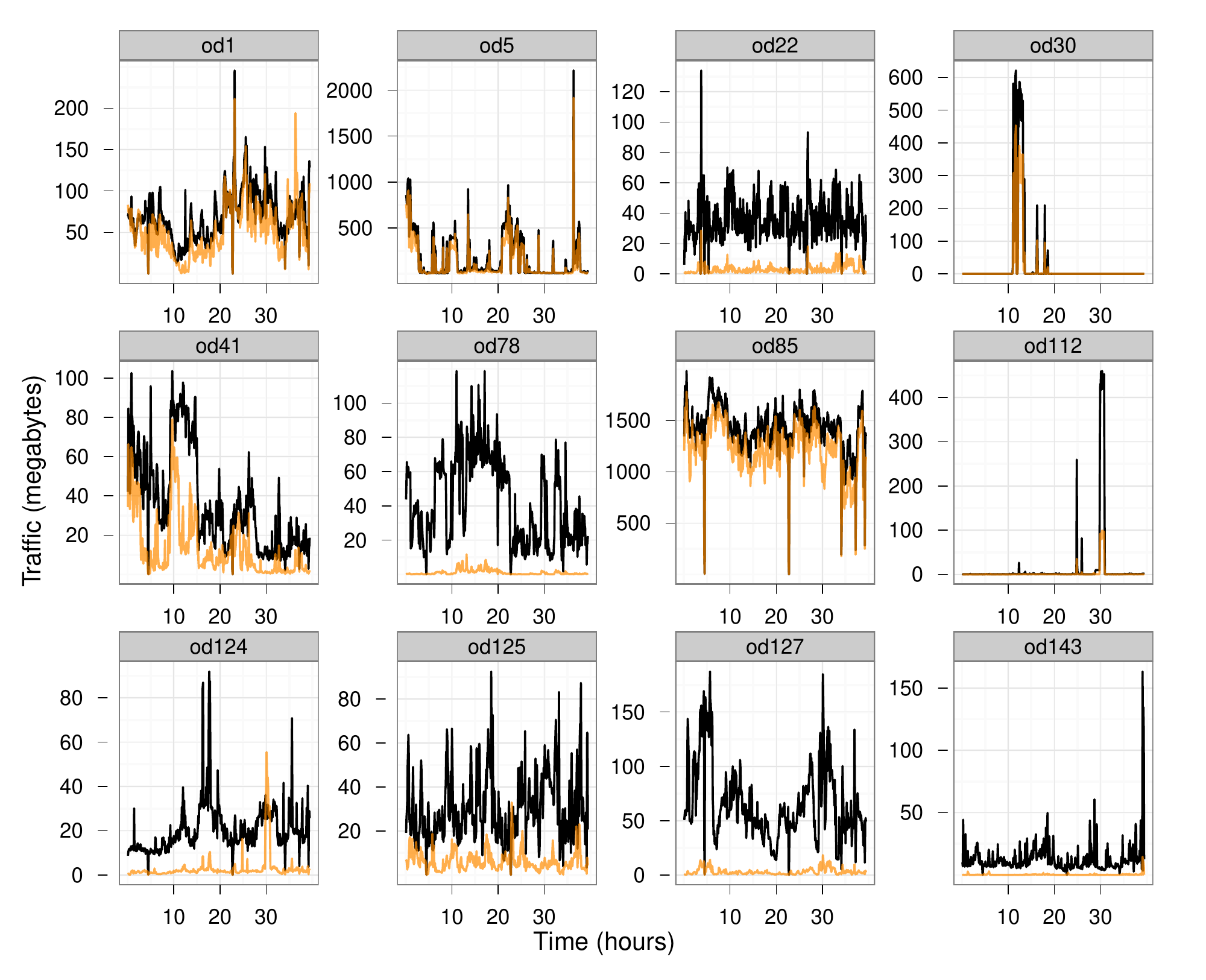}
\caption{Fitted values vs. ground truth for CMU data. Ground truth in black; Na\"ive prior in orange.}
\label{suppfig:naiiveCMU}
\end{figure*}

\end{document}